\documentclass[11pt]{article}
\usepackage[paper=a4paper,top=2cm,bottom=2.2cm,inner=2cm,outer=2cm]{geometry}

\usepackage[english]{babel}
\usepackage{inputenc}
\usepackage[T1]{fontenc}
\usepackage{lmodern}							
\usepackage[usenames,dvipsnames,svgnames,table]{xcolor}	
\usepackage{paralist}							
\usepackage[shortlabels]{enumitem} 				

\usepackage{verbatim} 
\usepackage{spverbatim} 					
\usepackage{listings}           

\usepackage{graphicx}                            
\usepackage{rotating}								
\usepackage[normalsize]{caption}	
\usepackage{subcaption}								
\usepackage{tabulary,tabularx,tabu}					
\usepackage[table]{xcolor}
\usepackage{adjustbox}				
\usepackage[section]{placeins}		
\usepackage{pgfplots}
\usepackage{tikz}

\usepackage{mathtools}                         
\usepackage{amsmath}         	
\usepackage{amsfonts}                 
\usepackage{amssymb}                  
\usepackage{dsfont,mathrsfs}                 
\usepackage{bm}							
\usepackage{slantsc}					
\usepackage[makeroom]{cancel}			

\usepackage{titletoc}					
\usepackage[nottoc,notbib]{tocbibind}	
\RequirePackage{secdot}     
\usepackage{appendix}					

\usepackage{amsthm}              
\usepackage{xspace}                
\usepackage{pdflscape}                 
\allowdisplaybreaks                         
\usepackage[shortcuts]{extdash}    
\usepackage{ifthen}				
\usepackage{pdfpages}			

\usepackage{setspace}				
\pgfplotsset{compat=1.16}

\usepackage[citecolor=blue,colorlinks=true,linkcolor=blue]{hyperref}

\usepackage[nameinlink]{cleveref}		

\newtheorem{dummy}{Dummy}[section]              
\newtheorem{proposition}[dummy]{Proposition}
\Crefname{proposition}{Proposition}{Propositions}

\Crefname{lemma}{Lemma}{Lemmas}
\newtheorem{theorem}[dummy]{Theorem}
\Crefname{theorem}{Theorem}{Theorems}

\theoremstyle{definition}

\newtheorem{example}[dummy]{Example}
\newtheorem{remark}[dummy]{Remark}

\newcommand*{\PoW}{{PoW}}
\newcommand*{\PPS}{{PPS}}
\newcommand*{\FPPS}{{FPPS}}

\newcommand \E{\mathbb{E}}

\newcommand{\vh}{\widehat{V}}
\newcommand{\ph}{\widehat{\psi}}

\newcolumntype{Y}{>{\centering\arraybackslash}X}
\newcolumntype{C}[1]{>{\centering\let\newline\\\arraybackslash\hspace{0pt}}p{#1}}

\usepackage{xparse}

\dottedcontents{section}[3em]{}{2em}{1pc}
\dottedcontents{subsection}[5em]{}{2.4em}{1pc}

\title{Blockchain mining in pools: Analyzing the trade-off between profitability and ruin}

\author{
	Hansj\"{o}rg Albrecher\footnote{The Faculty of Business and Economics, University of Lausanne and Swiss Finance Institute, Quartier UNIL-Chamberonne B\^{a}timent Extranef, 1015 Lausanne, Switzerland, \texttt{hansjoerg.albrecher@unil.ch}}
		\and
Dina Finger\footnote{The Faculty of Business and Economics, University of Lausanne, Quartier UNIL-Chamberonne B\^{a}timent Extranef, 1015 Lausanne, Switzerland, \texttt{dina.finger@unil.ch}}
	\and
	Pierre-O. Goffard\footnote{Laboratoire de Sciences Actuarielle et Financière EA2429, Université Claude Bernard Lyon 1, Universit\'e de Lyon. Institut de Science Financière et d'Assurances, 50 Avenue Tony Garnier, 69007 Lyon, France, \texttt{pierre-olivier.goffard@univ-lyon1.fr}}
}

\date{}

\begin{document}

\maketitle

\begin{abstract}
The resource-consuming mining of blocks on a blockchain equipped with a proof of work consensus protocol bears the risk of ruin, namely when the operational costs for the mining exceed the received rewards. In this paper we investigate to what extent it is of interest to join a mining pool that reduces the variance of the return of a miner for a specified cost for participation. Using methodology from ruin theory and risk sharing in insurance, we quantitatively study the effects of pooling in this context and derive several explicit formulas for quantities of interest. The results are illustrated in numerical examples for parameters of practical relevance. 
\end{abstract}


\section{Introduction}\label{sec:introduction}
A blockchain is a decentralized data ledger maintained by a \textit{Peer-to-Peer} network. Blockchain users issue transactions to the network peers who agree on those to be recorded by following a consensus protocol. In public and permissionless blockchains, such as the one for Bitcoin, the consensus protocol is called \textit{Proof-of-Work} (\PoW). The nodes, referred to as miners, compete to solve a challenging cryptographic puzzle using some brute force search algorithm. The first miner to come up with a solution includes the pending transactions in a block and is rewarded with newly minted crypto-coins. This reward compensates the operational cost of mining mainly induced by the consumption of electricity. The \PoW\,protocol is designed to be incentive compatible in the sense that a miner is compensated proportionally to her computational effort. When the {Peer-to-Peer} network grows large, the share of the network computing power owned by a given miner shrinks, which in turn makes the gains infrequent. The constant operating costs therefore endanger the solvency of miners and has led them to join forces by forming mining pools.\\  

\noindent Mining pools grant miners a steadier income, as block finding rewards are collected more often. The earnings are then fairly distributed to the pool participants based on their contribution to the computational effort. The simplest way to do so consists in splitting the reward whenever a block is found. This is the \textit{proportional} reward system. More sophisticated reward schemes have been put together to increase the amount of risk transferred from the miners to the pool and to fill the gaps of the proportional system that we will discuss later. These more sophisticated systems require the supervision of a manager who undertakes part of the risk in exchange for a commission. An early work of Rosenfeld \cite{DBLP:journals/corr/abs-1112-4980} provides a detailed overview on mining pool reward systems, see also the recent survey of Zhu et al.\ \cite{Zhu_2018}. In practice, the individual contribution of a miner is measured through a \textit{share} submission process. A \textit{share} refers to an easier-to-find 'fake' solution to the crypto-puzzle that miners must send to the pool manager to prove their involvement (for instance, a solution to the crypto-puzzle with only $m$ instead of the $n>m$ leading zeroes required for the successful mining of a block). In this work we provide a risk analysis of \textit{Pay-per-Share} (\PPS) reward systems in which the pool manager pays for each share submitted by the miners. In that way the manager takes on much of the randomness associated to the mining activity, which is therefore very appealing to the participant. Using utility theory, Schrijvers et al.\ \cite{Schrijvers2017} showed that such systems are incentive-compatible for risk-averse miners. Both Rosenfeld \cite{DBLP:journals/corr/abs-1112-4980} and Zhu et al.\ \cite{Zhu_2018} stressed that a scheme of this kind must go hand in hand with a proper capital allocation strategy on the part of the manager to avoid ruin. The aim of this paper is to provide risk-analytic tools to inform the decision making process of miners and pool managers. This is achieved by taking an approach inspired from insurance risk theory.\\

\noindent The wealth of miners and pool managers is modelled via stochastic processes that take into account operational costs, pool participation fees and the rewards for {share} and block finding. The resulting processes are similar to those appearing in the surplus modelling for insurance companies which collect premiums continuously and have to pay loss reimbursements to policyholders in case a claim occurs. A standard risk measure in this context is the ruin probability defined as the probability that the wealth process falls below zero, see e.g.\ Asmussen and Albrecher \cite{asmussen2010ruin} for an overview. This analogy was already used in Albrecher and Goffard \cite{albrecher2020profitability}, where the opportunity for miners to deviate from the prescribed protocol by withholding blocks was investigated. A first result was also obtained there in relation to the advantage of joining a mining pool which applies the proportional system. Our objective in this paper is to considerably extend this line of thinking towards the \textit{Pay-per-Share} redistribution systems that are  more commonly used in practice. We will also consider a variant of the model in which the collected rewards are random variables. This assumption will enable the application of classical results from double-sided jumps in a risk reserve process for modelling insurance portfolios, see for example Albrecher at al. \cite{albrecher2010direct}, Labb\'e and Sendova \cite{labbe2009expected}. Incorporating random rewards allows us to account for the transactions fees and the exchange rate of cryptocurrencies to fiat ones. Transaction fees are included by blockchain users to entice the network to process their transactions, see Easley al.\ \cite{easley2019mining} and Kasahara and Kawahara \cite{Kasahara2019}. The redistribution of the revenue generated by the transaction fees among the pool participants also varies from one mining pool to another. \\

\noindent A major concern associated to mining pool formation is the centralization of the network. Cong et al. \cite{Cong2020} have explained that miners who direct their mining power to multiple small mining pools enjoy the same risk sharing benefits as miners that choose to join a single mining pool. Hence the intuition that a larger mining pool would grow even larger is misguided. Empirical data shows that the participation fees are greater in larger mining pools, which naturally slows down their growth. We aim at providing more insight on the risk of centralization in the light of our analysis.\\     

\noindent The remainder of the paper is organized as follows. Section \ref{sec:reward_systems} provides an overview of the existing reward systems and describes the  \textit{Pay-Per-Share} mechanism in more detail, as it will be the focus later on. Formulas for the ruin probability and expected surplus for the pool manager are derived for deterministic rewards in Section \ref{secd} and for randomized reward in Section \ref{secr}. Section \ref{sec:minerdet0} provides formulas from the individual miner's perspective. Section \ref{seci} is devoted to numerical illustrations where the sensitivity of the risk and performance indicator is analysed with respect to the model parameters. Section \ref{secc} concludes.

\section{Risk models and reward systems}\label{sec:reward_systems}
A risk model defines the wealth of some company or individual as a stochastic process
\[
R_t = u - C_t + B_t,\text{ }t\geq0. 
\]  
which corresponds to the income $(B_t)_{t\geq0}$ net of the expenses $(C_t)_{t\geq0}$. The surplus process $(R_t)_{t\geq0}$ starts at some initial level $R_0 = u>0$. We take a continuous time approach where $t\in\mathbb{R}_+$, and $(C_t)_{t\geq0}$ and $(B_t)_{t\geq0}$ define increasing functions or stochastic processes. A risk analysis is relevant only if at least one of the model components is random. The activity of the company is profitable if on average the earnings exceed the expenses, namely $\mathbb{E}(B_t)>\mathbb{E}(C_t)$. Even if the net profit condition holds, the variability of the process $(R_t)_{t\geq0}$ can lead to bankruptcy as it may become negative. Define the ruin time as 
$$
\tau_u = \inf\{t\geq0:\;R_t < 0\},
$$  
which corresponds to the first time at which the surplus goes below $0$. The risk of bankruptcy is classically assessed by computing the ruin probability up to time $t\geq0$ defined by 
\begin{equation}\label{eq:ruin_probability}
\psi(u, t) = \mathbb{P}(\tau_u \leq t).
\end{equation}
It is sometimes more convenient from a mathematical point of view to consider the infinite-time horizon by letting $t\rightarrow\infty$, and in that case we write $\psi(u):=\lim_{t\to\infty}\psi(u, t)$. Following the rationale developed in \cite{albrecher2020profitability}, we also consider a performance indicator defined as 
\begin{equation}\label{eq:expected_surplus}
V(u, t) = \mathbb{E}(R_t\mathbb{I}_{\tau_u >t} ),
\end{equation}
which corresponds to the expected surplus at time $t\geq0$ in case ruin did not occur until then.\\

\noindent Consider a network of $n$ miners, where miner $i\in\{1,\ldots, n\}$ owns a share $p_i\in(0,1)$ of the network hashpower, i.e.\ $\sum_{i=1}^{n}p_i=1$. If the number of blocks found by the network is governed by a homogeneous Poisson process $(N_t)_{t\geq0}$ with intensity $\lambda$, then the number of blocks found by miner $i$ is a (thinned) Poisson process $(N^i_t)_{t\geq0}$ with intensity $p_i\cdot \lambda$. Denote by $b>0$ the amount of the reward for finding a new block and assume that the cumulative operational cost is a linear function with slope $c_i>0$ which depends on the price of the electricity and the computing power of miner $i$. The surplus process of miner $i$ is then given by
\begin{equation}\label{eq:miner_surplus_solo_mining}
R^i_t = u - c_i \cdot t + N^i_t\cdot b,\text{ }t\geq0. 
\end{equation}  
Model \eqref{eq:miner_surplus_solo_mining} has been considered by Albrecher and Goffard \cite{albrecher2020profitability}, and formulas for both the finite-time ruin probability \eqref{eq:ruin_probability} and the expected surplus  \eqref{eq:expected_surplus} were derived. To make the formulas more amenable for numerical evaluation, the authors then decided to approximate the fixed time horizon $t\geq0$ by an exponential random variable $T\sim\text{Exp}(t)$ with mean $t\geq 0$, resulting in tractable expressions for 
\begin{equation}\label{eq:rp_ep_proxy}
\ph(u,t)=\E[\psi(u,T)],\text{ and } \vh(u,t):= \E[V(u,T)],
\end{equation}
which were then used to carry out a numerical analysis.\\

\noindent Consider now a situation where a subset of miners $I\subset\{1,\ldots, n\}$ decides to gather in a mining pool. The cumulated hashpower of this pool is then
$$
p_I = \sum_{i\in I}p_i,
$$
and the arrival rate of block rewards for a given miner $i$ rises from $p_i\cdot\lambda$ to $p_I\cdot\lambda$. Because the reward is shared among the pool participants, the size of the reward collected by miner $i$ decreases from $b$ to $p_i\cdot b$. The expected surplus is the same when mining solo and mining for a pool, but the variance (and therefore the risk) is smaller when mining for a pool. The management of a mining pool relies heavily on the reward distribution mechanism set up by a pool manager. For the redistribution system to be fair, each miner must be remunerated in proportion to her calculation effort. Miner $i$ must earn a share $p_i/p_I$
of the mining pool total income. The pool manager has to find a way to estimate the contribution of each pool participant. This is done by submitting \textit{shares} which are partial solutions to the cryptopuzzle easier to find than the actual solution. The manager's cut is a fraction $f\in(0,1)$ of the block discovery reward $b$. We start by presenting the proportional reward system in Section \ref{ssec:proportional}.

\subsection{The proportional reward system}\label{ssec:proportional}
The proportional reward system splits time in \textit{rounds} which correspond to the time elapsed between two block discoveries. During these \textit{rounds}, the miners submit \textit{shares}. The ratio of the number of \textit{shares} submitted by miner $i$ over the total number of \textit{shares} submitted by her fellow mining pool participants determines her share of the reward and should converge to her share of the mining pool computing power, that is $p_i/p_I$ (for sufficiently low complexity of the shares, the latter limit will be a very good approximation for the actual situation indeed). The surplus of miner $i$ is then given 
\begin{equation}\label{eq:surplus_miner_proportional}
R_t^i = u - c_i\cdot t + N^I_t\cdot (1-f)\cdot \frac{p_i}{p_I}\cdot b,\text{ }t\geq0, 
\end{equation}
where $(N^I_t)$ is a Poisson proccess of intensity $p_I\cdot\lambda$ that gives the number of blocks appended to the blockchain by the mining pool. The duration of a \textit{round} is exponentially distributed $\text{Exp}\left[(p_I\lambda)^{-1}\right]$. The uncertainty on the length of the round has undesirable consequences on the time value of the \textit{shares} submitted by the miners. Indeed, if $n$ shares are submitted during a round, then the value of a given \textit{share} is $(1-f)\cdot b / n$. The longer a \textit{round} lasts, the greater the value of $n$ is. The \textit{shares} are worth less in longer rounds which triggers an exodus behavior of miners toward mining pools with shorter rounds. This phenomenon, called pool hopping, has been documented in the early work of Rosenfeld \cite{DBLP:journals/corr/abs-1112-4980}. Yet another drawback is that a miner that has found a full solution may delay the submission until her ratio of \textit{shares} submitted reflects her fraction of the mining pool computing power. The proportional system is not \textit{incentive-compatible} using the terminology of Schrijvers et al. \cite{Schrijvers2017}. A discounting factor may be applied to compensate the decreasing value of shares over time, see for instance the slush's method \cite{slush}.
\\

\noindent Our work is also concerned about the risk undertaken by pool managers. Within the frame of the proportional reward system, the surplus of the pool manager is given by
\begin{equation}\label{eq:surplus_pool_manager_proportional} 
R_t^I = u + N^I_t\cdot f\cdot b,\text{ }t\geq0.
\end{equation}
Model \eqref{eq:surplus_pool_manager_proportional} does not account for any mining pool operating cost. The mining costs are entirely borne by miners and the mining pool manager only serves as coordinator. A proportional-type reward system should therefore lead to a low management fee $f$. \\

Although this system provides fairness, it has weaknesses that justify the introduction of a more sophisticated distribution mechanism. In particular, if miners seek to actually transfer some of the risk associated to the mining activity to the pool manager, then they should rather turn to a mining pool based on a \textit{Pay-per-Share} system, which is the focus of this paper and introduced in the next section.

\subsection{The Pay-Per-Share reward system}\label{ssec:PPS}
In a \textit{Pay-per-Share} reward system, the pool manager immediately rewards the miners for each \textit{share} submitted. Let $(M_t)_{t\geq0}$ be a Poisson process of intensity $\mu$ that counts the number of \textit{shares} submitted by the entire network of miners up to time $t\geq0$. Denote by $q\in(0,1)$ the relative difficulty of finding a block compared to finding a share. Let $0<w<b$ be the reward for finding a \textit{share}. The number of \textit{shares} submitted by miner $i$ is then a (thinned) Poisson process $(M^i_t)_{t\geq0}$ and her surplus when joining a \PPS\, mining pool becomes  
\begin{equation}\label{eq:surplus_miner_pps}
R_t^i = u - c_i\cdot t + M^i_t\cdot w,\text{ }t\geq0. 
\end{equation}
The intensities of the processes $(N_t)_{t\geq0}$ and $(M_t)_{t\geq0}$ are linked through $\lambda  = q\cdot\mu$. By setting $w=(1-f)\cdot b\cdot q$, we observe that the surplus \eqref{eq:surplus_miner_proportional} and \eqref{eq:surplus_miner_pps} have the same expectation at time $t$, but the variance and therefore the risk associated to \eqref{eq:surplus_miner_pps} is lower. This reward system has been shown to be resistant to pool hopping and is incentive compatible. It also entails a significant transfer of risk to the pool manager whose surplus process is now given by
\begin{equation}\label{eq:surplus_manager_pps}
R^I_t = u - M_t^I\cdot w + N_t^I\cdot b,\text{ }t\geq0,
\end{equation}
making her subject to the risk of bankruptcy.

\begin{remark}
Since the process $(M_t^I)_{t\geq 0}$ requires solving for a problem of lower complexitiy than $(N_t^I)_{t\geq 0}$, ($N_t^I)_{t\geq 0}$ is a subset of the path defined by the process $(M_t^I)_{t\geq 0}$. It means that both processes are not independent. Concretely, at the moment of the block reward payment $b$, at the same time there is a realisation of the miners' reward $w$. As we sometimes will need to isolate downward jumps without the simultaneous upward jump point, we define another process with a reduced intensity. We apply the superposition theorem (see e.g. \cite{kingman1992poisson}) to the Poisson process $M_t^I$ by redefining the down jump process as $(M_t^{I,d})_{t\geq 0} \sim Poisson(\mu_d)$, where $\mu_d = \mu-\lambda$.
\end{remark}

Figure \ref{surplus_process} represents sample paths of the surplus processes for an individual miner and the pool manager.

\begin{figure}[!ht]
  \centering
  \begin{subfigure}{0.49\textwidth}
    \includegraphics[width=250px]{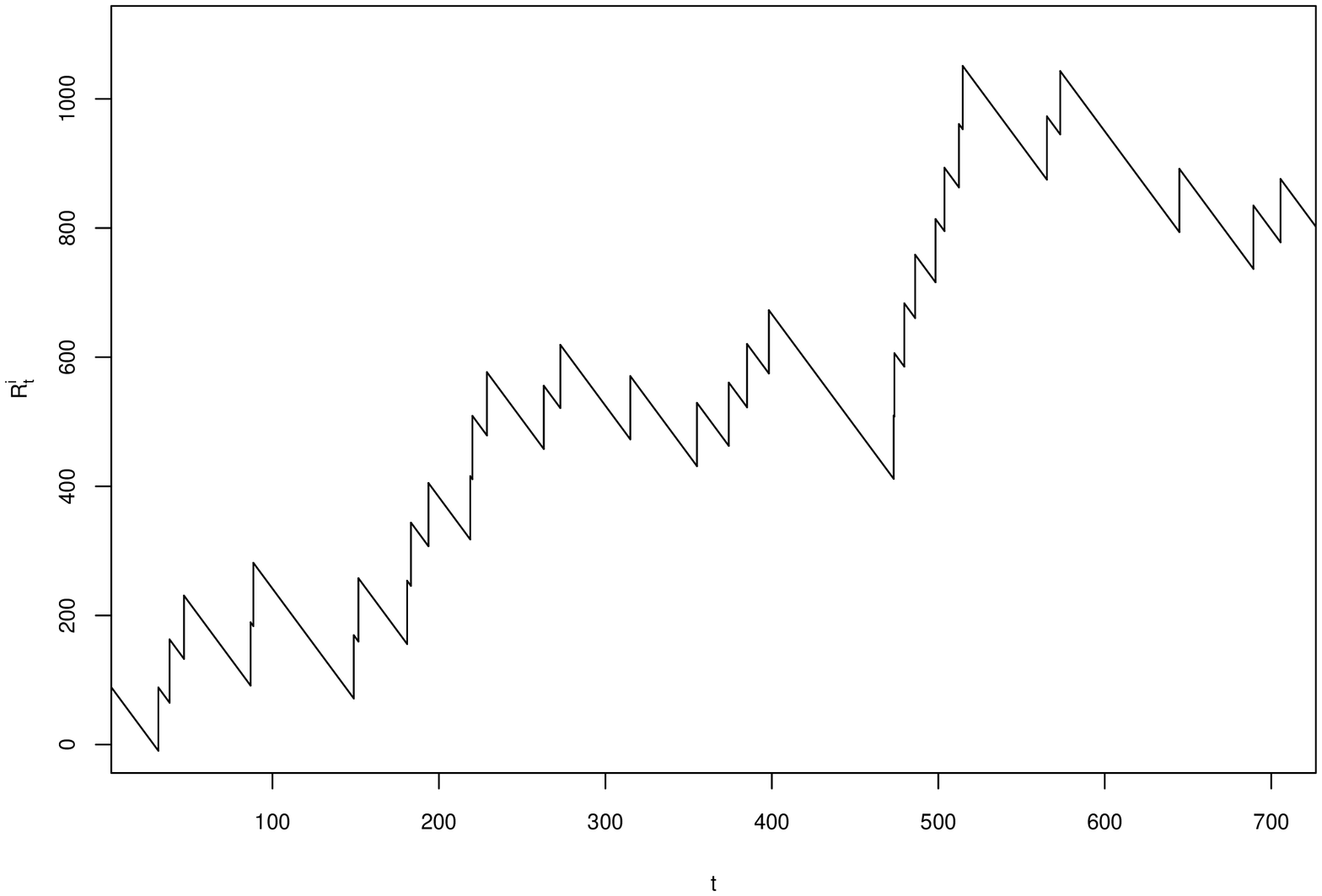}
    \caption{Individual miner}
  \end{subfigure}
  \begin{subfigure}{0.5\textwidth}
    \includegraphics[width=250px]{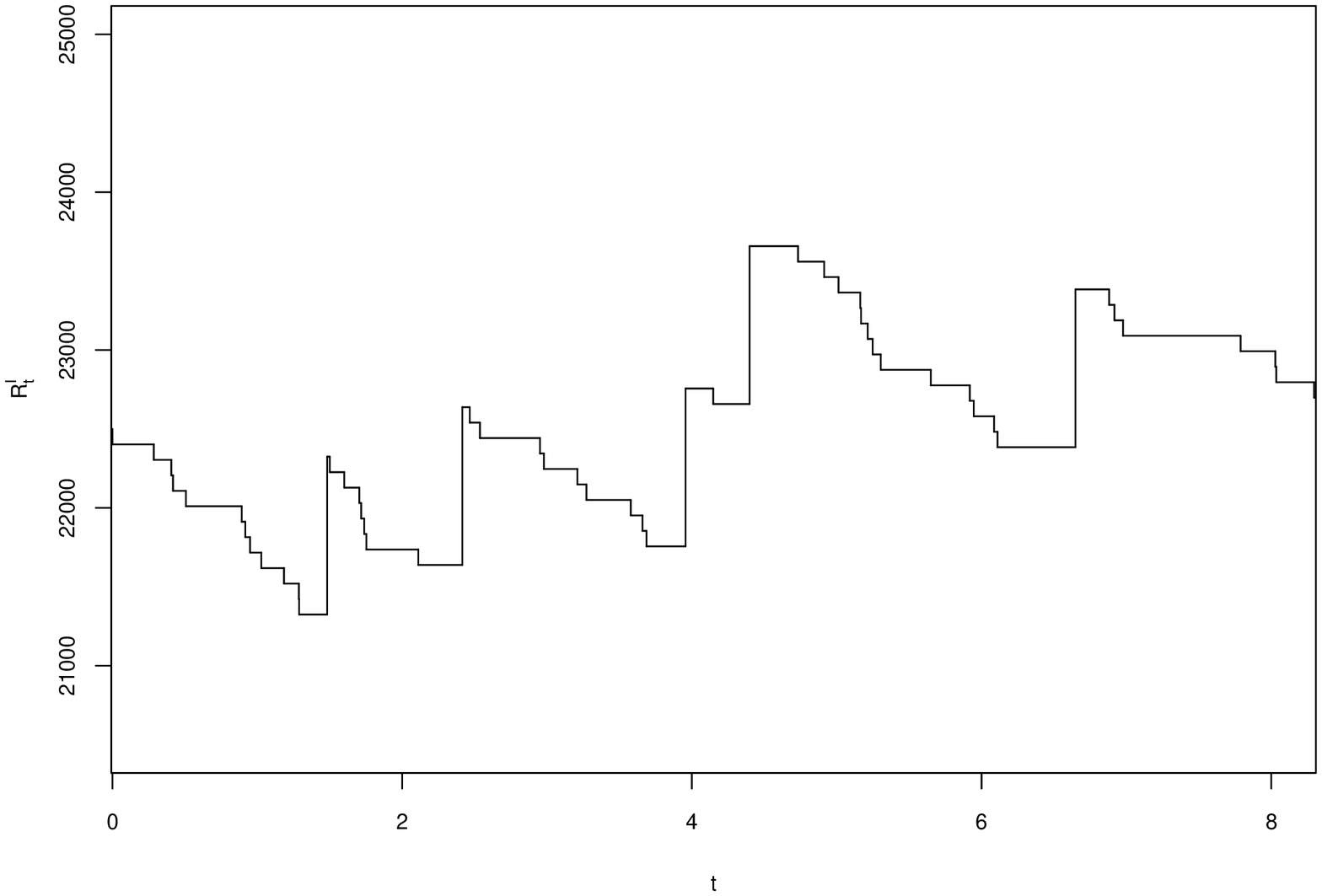}
    \caption{Pool manager}
  \end{subfigure}
  \caption{Illustration of surplus paths for the pool members and the pool manager.}
  \label{surplus_process}
\end{figure}

In addition to the bounty for finding a new block, blockchain users usually include a small financial incentive for the network to process their transaction. These transaction fees (e.g.\ referred to as \textit{gas} within the ETHEREUM blockchain), are known to be variable as they highly depend on the network congestion at a given time. Note also that since the operational cost is paid by miners using a fiat currency, it would be more accurate to account for the exchange rate of the cryptocurrency to some fiat currency. We can therefore model the successive rewards for \textit{shares} and blocks as sequences of nonnegative random variables denoted by $(W_k)_{k\geq1}$ and $(B_k)_{k\geq1}$ respectively, which for simplicity we will both assume to be \textit{i.i.d.}\ in this paper. A reward system that features a \textit{Pay-per-Share} mechanism and includes in the miners' reward the transaction fees is referred to as a \textit{Full Pay-per-Share} reward system by practitioners. The surplus of miner $i$ in a mining pool applying the \FPPS\,system is given by  
\begin{equation}\label{eq:surplus_miner_fpps}
R_t^i = u - c_i\cdot t + \sum_{k = 1}^{M^i_t}W_k,\text{ }t\geq0,
\end{equation}
and the surplus of the pool manager then becomes 
\begin{equation}\label{eq:surplus_manager_fpps}
R^I_t = u - \sum_{k = 1}^{M^I_t}W_k + \sum_{l = 1}^{N_t^I}B_l,\text{ }t\geq0. 
\end{equation}
In the following sections, we will now derive formulas for the ruin probability and expected surplus in case ruin did not occur up to a given time horizon for the models discussed above.

\section{Pool analysis with deterministic rewards}\label{secd}
We start with a fixed time horizon. For simplicity, we drop the superscript $I$ in the following developments.
\subsection{Deterministic time horizon}
For the pool manager's side, we first define some measures of interest. Let $\tau = \inf\{t\geq 0 : R_t < 0\}$ be the time of ruin of the pool manager, i.e. the first time his surplus reaches 0. The corresponding ruin probabilities in finite and infinite horizon respectively are given by
\begin{equation}
\psi(u,t)=\mathbb{P}(\tau\leq t),\ \text{and}\ \psi(u)=\mathbb{P}(\tau<\infty).
\end{equation}

The net profit condition in this case translates to $\lambda b > \mu w$. It implies from \cite{asmussen2010ruin}, that $\psi(u)<1$. We also define the expected surplus at time $t$ given that ruin has not occurred up to time $t$:
\begin{equation} \label{eq_expected_surplus}
V(u,t) = \mathbb{E}(R_t\mathbb{I}_{\tau>t}).
\end{equation}

Note that ruin can only occur at discrete times when the process $(M_t^d)_{t\geq0}$ admits a jump. We can rewrite the ruin time $\tau$ as

\begin{equation}
\tau = \inf\{t\geq 0 ; M_t^d w > u + N_t(b-w)\}=\inf\{t\geq 0 ; M_t^d  > u/w + N_t(b-w)/w\}.
\end{equation}

Equivalently, we can rewrite it as
\begin{equation} \label{eq_ruin_time}
\tau = \inf\{t\geq 0 ; M_t^d w/(b-w) > u/(b-w) + N_t\}
\end{equation}

to isolate the $(N_t)_{t\geq0}$ process with unit jumps. The study of the p.d.f.\ $f_{\tau}$ of $\tau$ is analogous to the derivations in \cite{goffard2019fraud}. Since $(N_t)_{t\geq0}$ and $(M_t^d)_{t\geq0}$  are Poisson process, they enjoy the order statistic property. That is, given that $N_t=n$, the jump times $\{T_1,\dots , T_n\}$ of the process $N_t$ have the same distribution as the order statistics vector of a random variable having distribution $F_t(s) = s/t,\ 0\leq s \leq t$. Further, let $\{S_n^d,n\in \mathbb{N}\}$ be the sequence of arrival times associated with the process $(M_t^d)_{t\geq0}$. Its distribution function is denoted by $F_{S_n^d}(t)$ and its p.d.f.\ by $f_{S_n^d}(t)$. Denote by $\lceil x \rceil$ the ceiling function. Following Corollary 1 from \cite{goffard2019fraud}, we proceed from Equation \eqref{eq_ruin_time} and derive the next steps.
\begin{theorem} \label{th_f_ruin_time}
Let $(N_t,\ t\geq 0)$ and $(M_t^d,\ t\geq 0)$ be Poisson processes with intensities $\{\lambda,\mu_d\}$ respectively, the p.d.f.\ of $\tau$ is given by
\begin{equation}\label{eq:pdf_ruin_time_deterministic_time_horizon}
f_{\tau}(t)=  \sum_{n=0}^{+\infty}\mathbb{E}\left[\frac{(-1)^n}{t^n}\, G_n\left[0\rvert S_{v_0},\dots,S_{v_{n-1}}^d\right]\rvert S_{v_n}^d=t\right]  f_{S_{v_n}^d}(t)\mathbb{P}\left[N_t=n\right],
\end{equation}
where $(v_n)_{n\geq0}$ is a sequence of integers defined as $v_n = \lceil n(b-w)/w+u/w\rceil$, $n\geq0$, and $\left(G_n(\cdot\rvert\{\ldots\}\right)_{n\in\mathbb{N}}$ is the sequence of Abel-Gontcharov polynomials defined in Appendix \ref{sec:ag_polynomials}.
\end{theorem}
The proof is delegated to Appendix \ref{sec:proof_deterministic_time_horizon}. The expression of the ruin time p.d.f.\ \eqref{eq:pdf_ruin_time_deterministic_time_horizon} is not convenient for numerical purposes. The infinite series in \eqref{eq:pdf_ruin_time_deterministic_time_horizon} must be truncated, possibly to a high order to reach an acceptable level of accuracy. Also, the evaluation of the Abel-Gontcharov polynomials relies on the recurrence relationships \eqref{eq:RecursionAGPolynomials} which are known to suffer from numerical instabilities. Moreover, the conditional expectation with respect to $\{S_{v_n}^d = t\}$ itself requires the use of Monte Carlo simulations. Finally, a similar algebraic expression for $V(u,t)$ is out of sight. In view of all these difficulties, we therefore propose as in \cite{albrecher2020profitability} a workaround which consists of replacing the deterministic time horizon by a random variable with exponential distribution.

\subsection{Exponential time horizon}
To obtain a nicer solution to the problem, we now randomize the time horizon $T$. The practical intuition suggests that the time horizon is never fixed in advance and is subject to various external factors, such as bitcoin price fluctuations, in-pool events etc. We choose the time horizon $T$ to be exponentially distributed with rate ${1}/{t}$ (so that $\mathbb{E}(T)=t$). This leads to computable expressions having an intuitive justification due to the lack of memory property of the exponential distribution. Let $\hat{V}(u,t):=\mathbb{E}(R_T\mathbb{I}_{\tau>T})$ denote the expected value of the surplus at the exponential time horizon $T$.

\begin{theorem}\label{th3.4}
Let $b$ and $w$, $b>w$, be fixed positive integers. Then the expected surplus at an exponential time horizon can be expressed in the form $$\hat{V}(u,t) =\sum_{i=1}^{w} c_i x_i^{u} + u+\lambda b\, t -(\lambda+\mu_d)w\, t,
$$
where $x_1,\ldots,x_w$ are the $w$ roots inside the unit disk of the equation
\begin{equation}\label{chare}
	\lambda x^{b}-(\lambda+\mu_d+{1}/{t})x^{w}+\mu_d =0,
\end{equation}
and the constants $c_1,\ldots,c_w$ are the solution of the linear equation system 
\begin{equation}\label{system}\left(\begin{array}{ccc}
		\lambda x_1^{b-w}-(\lambda+\mu_d+{1}/{t}) & \cdots &\lambda x_w^{b-w}-(\lambda+\mu_d+{1}/{t}) \\
		\lambda x_1^{b-w+1}-(\lambda+\mu_d+{1}/{t}) x_1 & \cdots &\lambda x_w^{b-w+1}-(\lambda+\mu_d+{1}/{t})x_w\\\vdots & \ddots &\vdots  \\
		\lambda x_1^{b-1}-(\lambda+\mu_d+{1}/{t}) x_1^{w-1} & \cdots &\lambda x_w^{b-1}-(\lambda+\mu_d+{1}/{t})x_w^{w-1}
	\end{array}\right)\left(\begin{array}{c}
		c_1 \\c_2\\\vdots\\c_w 
	\end{array}\right)=\left(\begin{array}{c}
		A_1 \\A_2\\\vdots \\A_w
	\end{array}\right),\end{equation}
with \[A_i=(i-1)\mu_d+\mu_dt(\lambda b-(\lambda+\mu_d)w)-\mu_dw,\quad i=1,\ldots,w.\]
\end{theorem}

\begin{proof}
Akin to the approach in \cite{albrecher2010direct}, consider some small $h>0$ and condition on the following scenarios during the time interval $(0,h)$: 
\begin{enumerate}
	\item no jump and $T>h$;
	\item no jump and $T\leq h$;
	\item occurrence of an upward jump;
	\item occurrence of a downward jump.
\end{enumerate}
All other combinations of these events have negligible probability in the limit $h\to 0$ that we will pursue below. One then obtains
\begin{equation}\label{hv}
\begin{split}
\hat{V}(u,t)&=e^{-(\frac{1}{t}+\lambda+\mu_d)h}\hat{V}(u,t) + \frac{1}{t}\int_0^h e^{-{s}/{t}}e^{-(\lambda +\mu_d) s} u\,ds\\
& + \lambda\int_0^he^{-\lambda s} e^{-({1}/{t}+\mu_d) s} \hat{V}(u+b-w,t)\,ds +\mu_d \int_0^he^{-\mu_d s} e^{-({1}/{t}+\lambda) s}  \hat{V}(u-w,t) \,ds.
\end{split} 
\end{equation}
We now take the derivative w.r.t. $h$ and set $h=0$ to obtain
\begin{equation} \label{eq_recursion_V_hat}
\lambda\hat{V}(u+b-w,t)-(\lambda+\mu_d+{1}/{t})\hat{V}(u,t)+\mu_d\hat{V}(u-w,t)+{u}/{t}=0,\quad u\ge 0.
\end{equation}
By definition of $\hat{V}(u,t)$ we have the boundary conditions $\hat{V}(u,t)=0$ for all $u<0$ and the linear boundedness $0\leq\hat{V}(u,t)\leq u+ (\lambda b-\mu_d w)t
$ in both $u$ and $t$ for all $u,t\ge 0$. 

Equation \eqref{eq_recursion_V_hat} is an inhomogeneous difference equation with constant coefficients (see e.g.\ \cite{jerri2013linear} for solution methods), which has 
the general solution 
\[\hat{V}(u,t)=\sum_{i=1}^{b} c_i x_i^{u} + d_0 + d_1 u
\]
with constants $\{c_i\}^b_{i=1}, \{x_i\}^{b}_{i=1}, d_0, d_1$ still to be determined.\\
Let us start with the inhomogeneous part: plugging the ansatz $d_0 + d_1 u$ into \eqref{eq_recursion_V_hat} gives
\[d_1=1,\;d_0=\lambda b t -(\lambda+\mu_d)w t.\]
For the homogeneous part, we consider the characteristic equation
\eqref{chare}, which by the Fundamental Theorem of Algebra has exactly $b$ complex roots $x_1,\dots,x_{b}$. The linear boundedness of $\hat{V}(u,t)$, however, excludes any solution with absolute value exceeding 1 (i.e., the corresponding constants $c_i$ must be zero). In fact, it turns out that exactly $w$ roots of the polynomial in \eqref{chare} are located inside the unit disk in the complex plane. To see this, observe first that 
$(\lambda+\mu_d+{1}/{t})x^{w}+\mu_d$ has exactly $w$ roots inside the unit disk (due to $\mu_d/(\lambda+\mu_d+{1}/{t})<1$). Then Rouch\'e's Theorem establishes that the same is true for the entire polynomial in \eqref{chare}, if 
\[\vert \lambda z^{b}\vert<\vert -(\lambda+\mu_d+{1}/{t})z^{w}+\mu_d\vert\; \text{on}\;\vert z\vert=1,\]
which translates into the condition 
\begin{equation}\label{ine}
\vert\mu_d-(\lambda+\mu_d+{1}/{t})z^{w}\vert>\lambda\; \text{on}\;\vert z\vert=1.
\end{equation}
The reverse triangle inequality states for any complex $a,b\in{\mathbb C}$ that $\vert a-b\vert\ge \Big\vert\vert a\vert-\vert b\vert\Big\vert$, which shows that for $\vert z\vert=1$ the left-hand side of \eqref{ine} is larger than $\lambda+1/t$, so that \eqref{ine} is indeed fulfilled.\\
It is now only left to determine the coefficients $c_1,\ldots,c_w$ corresponding to the $w$ roots $x_1,\ldots,x_w\in {\mathbb C}$ with $\vert x_i\vert<1$ of \eqref{chare}. To that end, note that \eqref{eq_recursion_V_hat} evaluated at $u=0,\ldots,w-1$ gives the following system of equations:
	\begin{align*}
	&\lambda\hat{V}(b-w,t)-(\lambda+\mu_d+{1}/{t})\hat{V}(0,t)=0,\\
	&\lambda\hat{V}(b-w+1,t)-(\lambda+\mu_d+{1}/{t})\hat{V}(1,t)+1/t=0,\\
	&\cdots \\
	&\lambda\hat{V}(b-1,t)-(\lambda+\mu_d+{1}/{t})\hat{V}(w-1,t)+(w-1)/t=0.
\end{align*}
Substituting the form 
\[\hat{V}(u,t)=\sum_{i=1}^{w} c_i x_i^{u} + u+a_t
\]
with $a_t=\lambda b\, t -(\lambda+\mu_d)w\, t$
into this system leads to 
\begin{align*}
	&\lambda\sum_{i=1}^{w} c_i x_i^{b-w}+\lambda(b-w)+\lambda a_t-(\lambda+\mu_d+{1}/{t})\left(\sum_{i=1}^{w} c_i+a_t\right)=0,\\
	&\lambda\sum_{i=1}^{w} c_i x_i^{b-w+1}+\lambda(b-w+1)+\lambda a_t-(\lambda+\mu_d+{1}/{t})\left(\sum_{i=1}^{w} c_ix_i+(1+a_t)\right)+1/t=0,\\
	&\cdots \\
	&\lambda\sum_{i=1}^{w} c_i x_i^{b-1}+\lambda(b-1)+\lambda a_t-(\lambda+\mu_d+{1}/{t})\left(\sum_{i=1}^{w} c_ix_i^{w-1}+(w-1+a_t)\right)+(w-1)/t=0.
\end{align*}
But the latter can be rewritten in the form \eqref{system}. 
\end{proof}

\begin{example}
Figure \ref{fig_Vh} depicts $\hat{V}(u,t)$ as a function of $u$ for the parameters $b = 100, w = 9, t = 1, \lambda=10, \mu_d = 90$. Note that for some capital levels $u$ the increase of $\hat{V}(u,1)$ from $u$ to $u+1$ is larger than for others. This is linked to how many down-jumps relative to up-jumps are needed to become negative, and due to the discrete nature of the problem such jumps in $\hat{V}(u,t)$ occur naturally. 
\begin{figure}[!ht]
		\includegraphics[width = 300px]{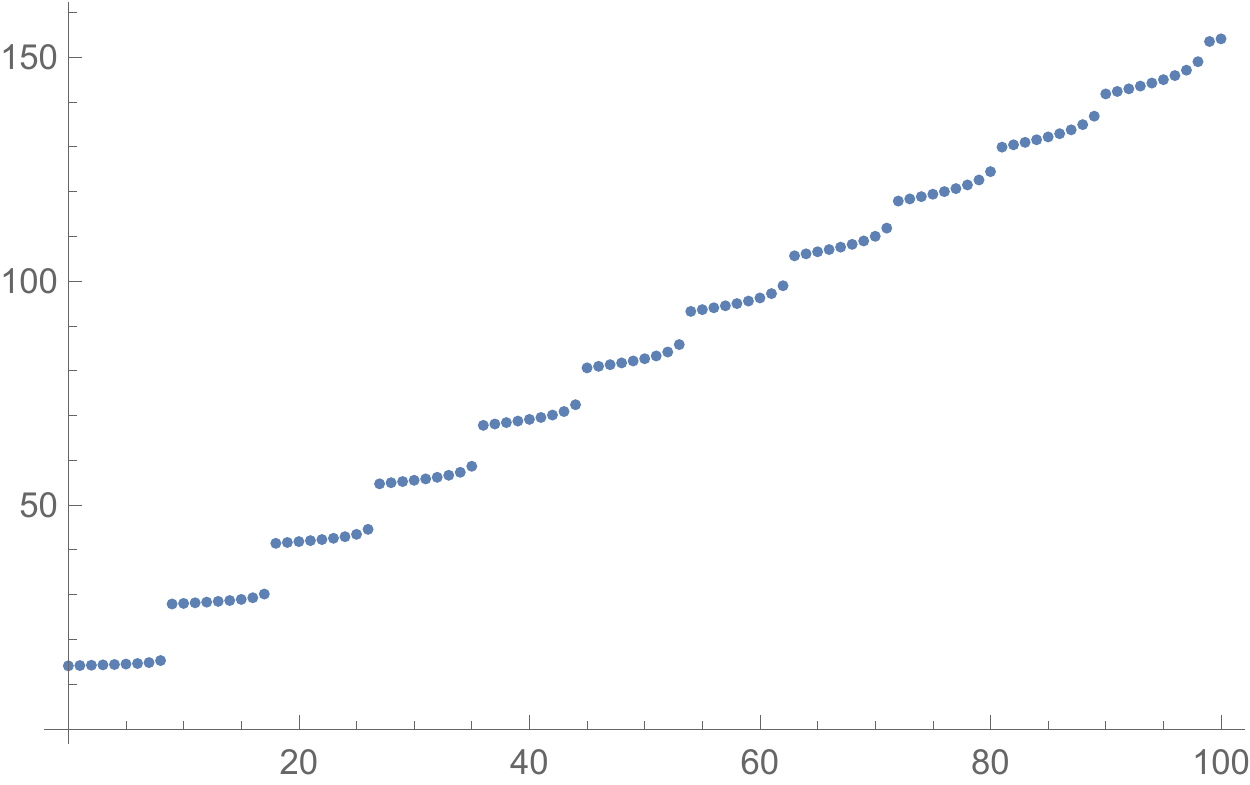}
		\centering
		\caption{$\hat{V}(u,1)$ as a function of $u$}
		\label{fig_Vh}
\end{figure}
\end{example}

In an analogous way, an explicit formula for $\widehat{\psi}(u,t)=\mathbb{E}\left[\psi(u,T)\right]$ can be derived.
\begin{theorem}
Let $b$ and $w$, $b>w$, be fixed positive integers. Then the ruin probability up to an exponential time horizon with mean $t$ is given by $$\widehat{\psi}(u,t) = \sum_{i=1}^{w} c_i x_i^{u}$$ 
where $x_1,\ldots,x_w$ are the $w$ roots inside the unit disk of Equation
\eqref{chare} and the constants $c_1,\ldots,c_w$ are the solution of the linear equation system \begin{equation}\label{systempsi}\left(\begin{array}{ccc}
		\lambda x_1^{b-w}-(\lambda+\mu_d+{1}/{t}) & \cdots &\lambda x_w^{b-w}-(\lambda+\mu_d+{1}/{t}) \\
		\lambda x_1^{b-w+1}-(\lambda+\mu_d+{1}/{t}) x_1 & \cdots &\lambda x_w^{b-w+1}-(\lambda+\mu_d+{1}/{t})x_w\\\vdots & \ddots &\vdots  \\
		\lambda x_1^{b-1}-(\lambda+\mu_d+{1}/{t}) x_1^{w-1} & \cdots &\lambda x_w^{b-1}-(\lambda+\mu_d+{1}/{t})x_w^{w-1}
	\end{array}\right)\left(\begin{array}{c}
		c_1 \\c_2\\\vdots\\c_w 
	\end{array}\right)=\left(\begin{array}{c}
		-\mu_d \\-\mu_d\\\vdots \\-\mu_d
	\end{array}\right).\end{equation}
\end{theorem}

\begin{proof}
We proceed in the same way as in the proof of Theorem \ref{th3.4}. The analogue of \eqref{hv} then is
\begin{equation}\label{hvpsi}
	\begin{split}
		\widehat{\psi}(u,t)=&\,e^{-(\frac{1}{t}+\lambda+\mu_d)h}\widehat{\psi}(u,t)+ \lambda\int_0^he^{-\lambda s} e^{-({1}/{t}+\mu_d) s} \widehat{\psi}(u+b-w,t)\,ds \\
		&  +\mu_d \int_0^he^{-\mu_d s} e^{-({1}/{t}+\lambda) s}  \widehat{\psi}(u-w,t) \,ds
	\end{split}
\end{equation}
and \eqref{eq_recursion_V_hat} is replaced by 
\begin{equation} \label{eq_recursion_psi}
	\lambda\widehat{\psi}(u+b-w,t)-(\lambda+\mu_d+{1}/{t})\widehat{\psi}(u,t)+\mu_d\widehat{\psi}(u-w,t)=0,\quad u\ge 0,
\end{equation}
which is the homogeneous equation of the former. The boundary conditions here are given by $\widehat{\psi}(u,t)=1$ for $u<0$ as well as the obvious bound $\widehat{\psi}(u+b-w,t)\le 1$ for all $u\ge 0$. Correspondingly, from the proof of the previous theorem we then know that 
\begin{equation}\label{formpsi}\widehat{\psi}(u,t)=\sum_{i=1}^{w} c_i x_i^{u} 
\end{equation}
with constants $c_1,\ldots,c_w$ still to be determined. Evaluating \eqref{eq_recursion_psi} at $u=0,\ldots,w-1$ gives
	\begin{align*}
	&\lambda\widehat{\psi}(b-w+j,t)-(\lambda+\mu_d+{1}/{t})\widehat{\psi}(j,t)+\mu_d=0,\quad j=0,\ldots,w-1.
\end{align*}
Substituting \eqref{formpsi} into these leads to 
\begin{align*}
	&\lambda\sum_{i=1}^{w} c_i x_i^{b-w+j}-(\lambda+\mu_d+{1}/{t})\left(\sum_{i=1}^{w} c_ix_i^j\right)+\mu_d=0,\quad j=0,\ldots,w-1,
\end{align*}
or equivalently \eqref{systempsi}.
\end{proof}

\section{Pool analysis with stochastic rewards}\label{secr}

Let us now assume that the up- and downward jumps in the dynamics of the pool manager's surplus  are stochastic. Under certain assumptions on the nature of these jumps, this will allow us to derive closed form formulas for $\widehat{\psi}$ and $\hat{V}$ in the spirit of \cite{albrecher2010direct}, see also  \cite[Ch.~4]{asmussen2010ruin}. Equation \eqref{eq:surplus_manager_pps} then is replaced by
\begin{equation} \label{eq_surplus_PM_stoch}
	R_t = u - \sum_{n=1}^{M_t^d} W_n + \sum_{n=1}^{N_t} B_{r,n},\quad t\geq 0,
\end{equation}
where we assume $W_n,\ n\in \mathbb{N}$ to be i.i.d.\ positive random variables with cumulative distribution function $F_W$ and finite mean 
representing payments to the pool members, and  $B_{r,n},\ n\in \mathbb{N}$ are assumed to be i.i.d.\ positive random variables with distribution function $F_{B_r}$ and finite mean 
representing the remaining inflow of bounty rewards diminished by the simultaneous payout to the respective pool member.

Consider the expected surplus of the pool manager as defined previously with a random time horizon $T$. Concretely, $T$ follows an exponential distribution with mean $t$. As in the previous section, we are interested in $\hat{V}(u,t)$. 

\begin{proposition}\label{propV} The quantity $\hat{V}(u,t)=\mathbb{E}(R_T\mathbb{I}_{\tau>T})$ for the pool surplus process \eqref{eq_surplus_PM_stoch} is a solution of the integral equation 
	\begin{equation} \label{inteq}
		\lambda\int_0^\infty\hat{V}(u+b_r,t)\,dF_{B_r}(b_r)-(\lambda+\mu_d+{1}/{t})\hat{V}(u,t)+\mu_d\int_0^u \hat{V}(u-w,t) \,dF_W(w)+{u}/{t}=0,\quad u\ge 0,
	\end{equation}
	with boundary conditions $\hat{V}(u,t)=0$ for all $u<0$ and $0\leq\hat{V}(u,t)\leq u+ (\lambda \mathbb{E}[B_r]-\mu_d \mathbb{E}[W])t$ 
	 for all $u,t\ge 0$.
\end{proposition}

\begin{proof}
	We extend the approach of the proof of Theorem \ref{th3.4} by conditioning on the size of the jump in case a jump occurs. For some small $h>0$ we correspondingly get
	\begin{equation}\label{neu0}
		\begin{split}
			\hat{V}(u,t)&=e^{-(\frac{1}{t}+\lambda+\mu_d)h}\hat{V}(u,t) + \frac{1}{t}\int_0^h e^{-{s}/{t}}e^{-(\lambda +\mu_d) s} u\,ds\\
			& + \lambda\int_0^he^{-\lambda s} e^{-({1}/{t}+\mu_d) s} \int_0^\infty\hat{V}(u+b_r,t)\,dF_{B_r}(b_r)\,ds\\
			&  +\mu_d \int_0^he^{-\mu_d s} e^{-({1}/{t}+\lambda) s}\int_0^u \hat{V}(u-w,t) \,dF_W(w)\,ds.
		\end{split}
	\end{equation}
Taking the derivative w.r.t.\ $h$ and setting $h=0$, one obtains \eqref{inteq}. The property $\hat{V}(u,t)=0$ for all $u<0$ follows by definition and the linear upper bound in $u$ and $t$ is obtained from the inequality $\hat{V}(u,t)\le \mathbb{E}(R_T)$. 

\end{proof}

For our purposes, it is very reasonable to assume (and will lead to simplified notation) that the generic random variables $B_r$ and $W$ are connected via
\begin{equation} \label{conn}
	B_r=a\,W
\end{equation}
for some constant $a>1$ that depends on the number of miners in the pool. Indeed, $W$ is the payment to the pool miner for solving a less complex puzzle, and $B_r$ can be seen as the bounty reward when the more complex puzzle is solved minus the payment to the miner who solved it, and that latter payment will be a constant fraction, depending on the specification of the pool rules. Note that for most positive random variables, a scaled version of it belongs to the same class of random variables with only the parameter(s) changed, and the latter is indeed the case for all distributional assumptions that we will pursue in this paper. In any case, all results below can easily be adapted to the case when $B_r$ and $W$ follow unrelated combinations of exponentials with different $n$ and $A_i$'s.\\

Let us now consider in more detail the case where both the up- and down-jumps are combinations of exponential random variables.  The latter class is dense in the class of all random variables on the positive half-line, so that the result is in fact quite general (see e.g.\ Dufresne \cite{dufresne2007fitting}). Concretely, the density of downward jumps is then assumed to be of the form 
\begin{equation} \label{combexp}
	f_W(w)=\sum_{i=1}^n A_i\alpha_i e^{-\alpha_i w},\quad w>0,
\end{equation}
where $\alpha_1<\alpha_2<\ldots<\alpha_n$ and $A_1+\cdots+A_n=1$ (but the $A_i$ are not necessarily positive). The Laplace transform of this density is given by
\[\tilde{f}_W(s):=\mathbb{E}(e^{-sW})=\sum_{i=1}^n A_i\,\frac{\alpha_i}{\alpha_i+s},\quad \text{Re}(s)> -\alpha_1.\]
From \eqref{conn}, we then have
\begin{equation} \label{combexpB}
	f_{B_r}(b_r)=\sum_{i=1}^n A_i\beta_i e^{-\beta_i b_r},\quad b_r>0
\end{equation}
with $\beta_i=\alpha_i/a$, $i=1,\ldots,n$.

\begin{theorem}\label{vhatcombexp}
	If $W$ and $B_r$ are combinations of exponential random variables with densities given in \eqref{combexp} and \eqref{combexpB}, then we have
	\begin{equation}\label{Vcombexp}
		\hat{V}(u,t) = \sum_{k=1}^n C_k e^{-r_k u}+u+t\,\sum_{i=1}^n A_i\left(\frac{\lambda}{\beta_i}-\frac{\mu_d}{\alpha_i}\right),
	\end{equation}
	where $r_1,\dots,r_n$ are the solutions with positive real parts of 
	\begin{equation} \label{VLund}
		\quad 
		\lambda \sum_{i=1}^n A_i \frac{\beta_i}{\beta_i+ r}+\mu_d \sum_{i=1}^n A_i \frac{\alpha_i}{\alpha_i -r} -(\lambda+\mu_d+1/t)=0
	\end{equation}
	and 
	\begin{equation}\label{VCk}
		\quad C_k= \frac{\sum_{j=1}^nB_j  \prod\limits_{h=1}^n(\alpha_j-r_h)\prod\limits_{i=1,i\neq j}^n\frac{r_k-\alpha_i}{\alpha_j-\alpha_i}}{\prod\limits_{h=1,h\neq k}^n(r_k-r_h)}, \quad k=1,\ldots,n	
	\end{equation}
	with $$B_j = \frac{1}{\alpha_j^2}-\frac{t}{\alpha_j}\sum_{i=1}^n A_i\left(\frac{\lambda}{\beta_i}-\frac{\mu_d}{\alpha_i}\right), \quad j=1,\ldots,n.$$
\end{theorem}

\begin{proof}
	Substituting \eqref{combexp} and \eqref{combexpB} into \eqref{inteq}, we get 
	\begin{multline*}	\lambda\,\sum_{i=1}^n A_i\beta_i\int_0^\infty\hat{V}(u+b_r,t) e^{-\beta_i b_r}db_r\\-(\lambda+\mu_d+{1}/{t})\hat{V}(u,t)+\mu_d\,\sum_{i=1}^n A_i\alpha_i \int_0^u \hat{V}(u-w,t) e^{-\alpha_i w}dw+{u}/{t}=0,\quad u\ge 0.
	\end{multline*}
	The function $\hat{V}(u,t)$ then has the form	\[\hat{V}(u,t) = \sum_{k=1}^n C_k e^{-r_k u}+d_1 u+d_0,\]
	for constants $C_1,\ldots,C_n,r_1,\ldots,r_n, d_0,d_1$ to be determined. In fact, plugging this ansatz into the above equation shows that  comparing coefficients of $e^{-r_k u}$ exactly gives \eqref{VLund} (which is a generalized Lundberg equation in the terminology of ruin theory, cf.\ \cite{asmussen2010ruin}). That equation has exactly $n$ solutions with positive real part $r_1,\ldots,r_n$ and $n$ solutions with negative real part (see e.g.\ \cite{zhang}). The solutions with negative real part would enter $\hat{V}$ with positive real part and are correspondingly irrelevant for our purpose, as that would violate the linear boundedness of the resulting $\hat{V}$ (in other words, the coefficients in front of such terms need to be zero). Comparing coefficients of $e^{-\alpha_i u}$, $i=1,\ldots,n$ gives
		\begin{equation}\label{VCk_0}
		\sum_{k=1}^n\frac{C_k}{\alpha_i-r_k}=\frac{d_1}{\alpha_i^2}-\frac{d_0}{\alpha_i}, \quad i=1,\ldots,n.
	\end{equation}
	Coefficients in front of $u\,e^{-\alpha_i u}$, $i=1,\ldots,n$ all cancel. After a little algebra, one sees that a comparison of coefficients of $u$ in that equation establish $d_1=1$ and  a comparison of the constant coefficients gives
	$$d_0=t\,\sum_{i=1}^n A_i\left(\frac{\lambda}{\beta_i}-\frac{\mu_d}{\alpha_i}\right).$$ 
	These obtained values of $d_1$ and $d_0$ can now be plugged into \eqref{VCk_0}, and the resulting system of linear equations can be solved explicitly to give \eqref{VCk}, either by realizing that the coefficient matrix is a Cauchy matrix or by using the trick of rational function representation developed in \cite[Sec.4]{albrecher2010direct}. 
\end{proof}

\begin{example} A particular simple example of the above is the case where $W$ is exponentially distributed with parameter $\alpha$ and $B_r$ is exponentially distributed with parameter $\beta$. In that case $n=1$ in Theorem \ref{vhatcombexp} and we obtain	
	\begin{equation}\label{Vcombexpe}
		\hat{V}(u,t) = \left(\frac{1}{\alpha^2}-\frac{t}{\alpha}\left(\frac{\lambda}{\beta}-\frac{\mu_d}{\alpha}\right)\right)  (\alpha-R) e^{-R u}+u+t\, \left(\frac{\lambda}{\beta}-\frac{\mu_d}{\alpha}\right),
	\end{equation}
	where $R$ is the (unique) solution with positive real part of 
	\begin{equation} \label{VLunde}
		\quad 
		\lambda  \,\frac{\beta}{\beta+ r}+\mu_d\,  \frac{\alpha}{\alpha -r} -(\lambda+\mu_d+1/t)=0.
	\end{equation}
\end{example}

Let us now move on to study the ruin probability $\widehat{\psi}(u,t)=\mathbb{E}\left[\psi(u,T)\right]$ in the present context.

\begin{theorem}\label{psihatcombexp}
	If $W$ and $B_r$ are combinations of exponential random variables with densities given in \eqref{combexp} and \eqref{combexpB}, then we have
	\begin{equation}\label{psicombexp}
		\widehat{\psi}(u,t) = \sum_{k=1}^n D_k e^{-r_k u},
	\end{equation}
	where $r_1,\dots,r_n$ are the $n$ solutions with positive real parts of \eqref{VLund}
	and 
	\begin{equation}\label{psiCk}
		\quad D_k= \frac{\sum_{j=1}^n\frac{1}{\alpha_j}\,  \prod\limits_{h=1}^n(\alpha_j-r_h)\prod\limits_{i=1,i\neq j}^n\frac{r_k-\alpha_i}{\alpha_j-\alpha_i}}{\prod\limits_{h=1,h\neq k}^n(r_k-r_h)}, \quad k=1,\ldots,n.
	\end{equation}
\end{theorem}

\begin{proof}
We can proceed in the same way as in the proof of Proposition \ref{propV} to derive an  integral equation for the ruin probability. The analogue of Equation \eqref{neu0} here is	\begin{equation}\label{psi_h}
	\begin{split}
		\widehat{\psi}(u,t)&=e^{-(\frac{1}{t}+\lambda+\mu_d)h}\widehat{\psi}(u,t) + \lambda\int_0^he^{-\lambda s} e^{-({1}/{t}+\mu_d) s} \int_0^\infty\widehat{\psi}(u+b_r,t)\,dF_{B_r}(b_r)\,ds\\
		&  +\mu_d \int_0^he^{-\mu_d s} e^{-({1}/{t}+\lambda) s}\left(\int_0^u \widehat{\psi}(u-w,t) \,dF_W(w)+\int_u^\infty 1 \,dF_W(w)\right)\,ds.
	\end{split}
\end{equation}
Taking the derivative w.r.t.\ $h$ and evaluating at $h=0$ then gives
\begin{equation} \label{psiinteq}
	\lambda\int_0^\infty\widehat{\psi}(u+b_r,t)\,dF_{B_r}(b_r)-(\lambda+\mu_d+{1}/{t})\widehat{\psi}(u,t)+\mu_d\int_0^u \widehat{\psi}(u-w,t) \,dF_W(w)+\mu_d(1-F_W(u))=0,\quad u\ge 0.
\end{equation}
Here the boundary conditions are $\widehat{\psi}(u,t)=1$ for $u<0$ and $\widehat{\psi}(u,t)\le 1$ for $u\ge 0$ and arbitrary $t>0$, and uniqueness of its solution follows analogously to Theorem \ref{vhatcombexp}. Under the assumptions on $F_{B_r}$ and $F_W$ this reads
\begin{multline} \label{psiinteqcombexp}
	\lambda\sum_{i=1}^n A_i\beta_i \int_0^\infty\widehat{\psi}(u+b_r,t)\,e^{-\beta_i b_r}\,db_r-(\lambda+\mu_d+{1}/{t})\widehat{\psi}(u,t)\\+\mu_d\sum_{i=1}^n A_i\alpha_i\int_0^u \widehat{\psi}(u-w,t) \, e^{-\alpha_i w}\,dw+\mu_d\sum_{i=1}^n A_ie^{-\alpha_i u}=0,\quad u\ge 0.
		\end{multline}
In analogy to the proof of Theorem \ref{vhatcombexp} we then see that the ruin probability must have the form
\[\widehat{\psi}(u,t) = \sum_{k=1}^n D_k e^{-r_k u}\]
for constants $D_1,\ldots,D_n$ to be determined, and $r_1,\ldots,r_n$ being the $n$ positive solutions of \eqref{VLund}. The constants $D_k$ are now obtained by substituting the above expression into \eqref{psiinteqcombexp} and comparing coefficients of $e^{-\alpha_i u}$, $i=1,\ldots,n$. This gives
\begin{equation}\label{VDk_0}
	\sum_{k=1}^n\frac{D_k}{\alpha_i-r_k}=\frac{1}{\alpha_i}, \quad i=1,\ldots,n.
\end{equation}
This system of linear equations is again of Cauchy matrix form with explicit solution \eqref{psiCk}, establishing the result.
\end{proof}

\begin{example} If $W$ and $B_r$ are exponentially distributed with parameter $\alpha$ and $\beta$, respectively, then \eqref{psicombexp} simplifies to	
	\begin{equation}\label{psiexpe}
		\widehat{\psi}(u,t) = (1-R/\alpha)  e^{-R u},\;u\ge 0,
	\end{equation}
	where $R$ is the (unique) solution with positive real part of \eqref{VLunde}. \\
	Note that for $t\to\infty$ one obtains $R=(\lambda\alpha-\mu_d\beta)/(\lambda+\mu_d)>0$, so that
	\begin{equation}\label{psiexpoinf}
	{\psi}(u) = \frac{\mu_d(1+\beta/\alpha)}{\lambda+\mu_d} e^{-\frac{\lambda\alpha-\mu_d\beta}{\lambda+\mu_d} u},\;u\ge 0.
\end{equation}
In particular, without initial capital in the pool, the infinite-time ruin probability amounts to 
\[{\psi}(0) = \frac{\mu_d(1+\beta/\alpha)}{\lambda+\mu_d} ,\]
in accordance with Formula (8.1) in \cite{albrecher2010direct}. 
\end{example}
%

\section{Individual miner analysis}\label{sec:minerdet0}

\subsection{Deterministic rewards}\label{sec:minerdet}
Comparing the formula describing the miner's surplus under the PPS pooling scheme \eqref{eq:surplus_miner_pps} with the solo-mining surplus \eqref{eq:miner_surplus_solo_mining}, one can see that they are in fact the same type of process, only distinguished by the reward amount and frequency. Correspondingly, the formulas obtained by Albrecher and Goffard \cite{albrecher2020profitability} for the expected value of the surplus and the ruin probability of a honest miner apply in the PPS case with deterministic rewards. Adapted to the present context, we hence get:
%
%
\begin{theorem}\cite{albrecher2020profitability}
For the miner's surplus process $R_t^i = u - c_i\cdot t + M^i_t\cdot w,\text{ }t\geq0$, with $M^i_t \sim Poisson(p_i\mu t)$, the value function $\hat{V}(u,t)$ can be expressed as

\begin{equation}
\hat{V}(u,t) = u + (p_i\mu w-c_i)t(1-e^{\rho^*u}),
\end{equation}

where $\rho^*$ is the negative solution of the equation 

\begin{equation} \label{eq:rho}
-c_i\rho+p_i\mu(e^{w\rho}-1)=1/t.
\end{equation}

\end{theorem}
%

\begin{theorem}\cite{albrecher2020profitability}
For the same surplus process, the ruin probability with exponential time horizon is given by $
\hat{\psi}(u,t) = e^{\rho^*u}$,
where $\rho^*$ is the negative solution of \eqref{eq:rho}.
\end{theorem}
%

\subsection{Stochastic rewards}\label{sec:minersto}

Consider now the same surplus process as in the previous section, but with stochastic rewards. Let us define this process by 

\begin{equation}
R_t^i = u - c_i\cdot t + \sum_{n=1}^{M_t^i} W_n,\text{ }t\geq0,
\end{equation}

where we assume $W_n,\ n\in \mathbb{N}$ to be i.i.d.\ positive random variables with cumulative distribution function $F_W$ and finite mean and $M^i_t \sim Poisson(p_i\mu t)$ as previously. This type of process is denominated as the \textit{dual problem} in the insurance context, see e.g.\ \cite{Avanzi2007}. We assume that the net profit condition $p_i\mu\mathbb{E}[W_n]>c_i$ is satisfied.\\
We are again interested in deriving the expected value of the surplus and the ruin probability for the miner. To simplify the computations, we consider again an exponential time horizon.
\begin{theorem}\label{th_Vminer}
For exponential time horizon, the expected value of the miner's surplus $\hat{V}(u,t)$ can be expressed as the solution of the integro-differential equation 
\begin{equation}\label{eq_mineride}
c_i\hat{V}'(u,t)+(\frac{1}{t}+p_i\mu)\hat{V}(u,t)-p_i\mu\int_0^{+\infty}\hat{V}(u+w,t)dF_W(w)-u/t=0,
\end{equation}
with boundary conditions $\hat{V}(0,t)=0$ and $0\leq \hat{V}(u,t) \leq u-c_it+p_i\mu\mathbb{E}[W]$.
\end{theorem}

\begin{proof}
As in previous sections, by conditioning the occurrence of $T$ to a small time interval $(0,h)$, we can write the value function as 
\begin{equation}
\begin{split}
\hat{V}(u,t) & = e^{-h(\frac{1}{t}+p_i\mu)}\hat{V}(u-c_ih,t) + \int_0^h \frac{1}{t} e^{-s(\frac{1}{t}+p_i\mu)}(u-c_is)ds\\
& + \int_0^h p_i\mu e^{-s(\frac{1}{t}+p_i\mu)} \int_0^{+\infty} \hat{V}(u-c_ih+w,t)dF_W(w)ds.
\end{split}
\end{equation}
Taking the derivative w.r.t.\ $h$ and evaluating it at $h=0$ gives us \eqref{eq_mineride}. The boundary condition follows from ruin considerations.
\end{proof}
For rewards being distributed as a combination of exponential random variables \eqref{combexp}, we can refine Theorem \ref{th_Vminer}.
\begin{theorem}
When $W$ has density $f_W(w)=\sum_{j=1}^n A_j\alpha_j e^{-\alpha_j w}$, $w>0$, then
\begin{equation}\label{eq_Vhatminercomb}
\hat{V}(u,t) = t\left(c_i-p_i\mu\sum_{j=1}^n\frac{A_j}{\alpha_j}\right)e^{-Ru}+u+t\left(p_i\mu\sum_{j=1}^n\frac{A_j}{\alpha_j}-c_i\right),\ u>0,
\end{equation}
where $R$ is the unique solution with positive real part of the equation
\begin{equation*}
c_iR+p_i\mu\sum_{j=1}^n\frac{A_j\alpha_j}{R+\alpha_j}-(\frac{1}{t}+p_i\mu)=0.
\end{equation*} 
\end{theorem}

\begin{proof}
Equation \eqref{eq_mineride} translates into
\begin{equation}\label{eq_mineridecomb}
c_i\hat{V}'(u,t)+(\frac{1}{t}+p_i\mu)\hat{V}(u,t)-p_i\mu\sum_{j=1}^n A_j\alpha_j \int_0^{+\infty}\hat{V}(u+w,t)e^{-\alpha_j w}dw-u/t=0.
\end{equation}
This equation has a solution of the form
\begin{equation}
\hat{V}(u,t)=C e^{-Ru}+d_1u+d_0
\end{equation}
and we plug this ansatz into \eqref{eq_mineridecomb}
\begin{equation}
\begin{split}
&c_i(-R C e^{-Ru}+d_1)+(\frac{1}{t}+p_i\mu)(C e^{-Ru}+d_1u+d_0)\\
&-p_i\mu\sum_{j=1}^n A_j\alpha_j \int_0^{+\infty}(C e^{-R(u+w)}+d_1(u+w)+d_0)e^{-\alpha_j w}dw-u/t=0.
\end{split}
\end{equation}
Comparing coefficients, we obtain
\begin{equation*}
d_1 = 1,\quad d_0 = t\left(p_i\mu\sum_{j=1}^n\frac{A_j}{\alpha_j}-c_i\right).
\end{equation*}
Further, a comparison of the coefficients in front of $e^{-Ru}$ simplifies to the following equation:
\begin{equation}
c_iR+p_i\mu\sum_{j=1}^n\frac{A_j\alpha_j}{R+\alpha_j}-(\frac{1}{t}+p_i\mu)=0.
\end{equation}
Similarly to the Lundberg equation derived in \cite{Lu2010}, we note that there is one positive root $R$ to this equation. To complete the proof, we consider the boundary condition $\hat{V}(0,t)=0$ and substituting into the ansatz gives $C=-d_0$.
\end{proof}

\begin{example}
When $W$ is exponentially distributed, i.e. $f_W(w)=\alpha e^{-\alpha w},\ w>0$, Equation \eqref{eq_Vhatminercomb} simplifies to 
\begin{equation}
\hat{V}(u,t) = t\left(c_i-\frac{p_i\mu}{\alpha}\right)e^{-Ru}+u+t\left(\frac{p_i\mu}{\alpha}-c_i\right),\ u>0,
\end{equation}
where $R$ is the solution with positive real part of 
\begin{equation*}
c_iR^2+(\alpha c_i-\frac{1}{t}-p_i\mu)R-\alpha\frac{1}{t}=0.
\end{equation*}
\end{example}
\begin{theorem}
For exponential time horizon, the miner's ruin probability can be expressed as
\begin{equation}
\hat{\psi}(u,t) = e^{-R\cdot u},
\end{equation}
where $R$ is the unique positive root of
\begin{equation}
p_i\mu + \frac{1}{t} -c_i R = p_i\mu \mathbb{E}[e^{-RW_n}].
\end{equation}
\end{theorem}

\begin{proof}
The proof is adapted from Example 2 of Mazza and Rulli{\`e}re \cite{Mazza2004}. From the latter, we have that the Laplace transform of the ruin time $\tau$ in the dual problem is $\mathbb{E}[e^{-s\tau}]=e^{-R(s)\cdot u}$, with $R(s)$ being the unique positive root of $p_i\mu + s -c_i R = p_i\mu \mathbb{E}[e^{-RW_n}]$. Since the ruin probability up to an exponential time horizon can be rewritten as 
\begin{equation}
\hat{\psi}(u,t) = \mathbb{E}[\mathbb{P}(T>\tau)\mid\tau],
\end{equation}
with $T\sim Exp(1/t)$, it immediately follows that 
\begin{equation}
\hat{\psi}(u,t) = \mathbb{E}[e^{\tau/t}]
\end{equation}
which completes the proof.
\end{proof}

\begin{example}
If $W$ is an exponential random variable, i.e. $f_W(w)=\alpha e^{-\alpha w},\; w>0$, then the ruin probability reduces to 
\begin{equation}
\hat{\psi}(u,t) = e^{-R^*u},
\end{equation}
where
\begin{equation}
R^* = \frac{{1}/{t}+p_i\mu-c_i\alpha+\sqrt{\Delta}}{2c_i},\ \Delta = (c_i\alpha -p_i\mu-{1}/{t})^2+4c_i\alpha/t.
\end{equation}
\end{example}

\begin{remark}
Results concerning the ruin probability can also be retrieved from the respective results for a more general renewal model  considered in Alcoforado et al. \cite{Alcoforado2021}.
\end{remark}

\section{Numerical illustration}\label{seci}
\subsection{Pool manager}
In this section, we will illustrate the pool dynamics in both the deterministic and stochastic setting. In addition, we will perform a sensitivity analysis on main decision variables from the pool's perspective.\\

First, let us define the set of parameters used in the following examples. For each illustration, we keep all the parameters fixed to these levels except the one that is varying : $t = 336,\ p_I = 0.1,\ q = 0.1,\ f = 0.02,\ b = 1000 MU,\ w = (1-f)bq = 98 MU,\ \lambda = 6p_I = 0.6,\ \mu_d = 6p_I(1/q-1) = 5.4$.

The units we use are hours ($h$) for the time parameters and monetary units ($MU$) for the value functions. The choice for the time horizon $t$ is equal to 2 weeks because it is linked to the period of difficulty adjustment. The monetary units are related to bitcoin in this way : $1000MU = 6.25 BTC$. The reason for this scaling is purely practical to solve the deterministic problem which involves integer constraints. As of May 28th 2021, $1BTC\approx \$35670.5$, so $1MU\approx \$231.85$.\\

Figure \ref{fig_V} compares the function $\hat{V}(u,t)$ defined in Theorem \ref{th3.4} with the Monte Carlo simulation of the mining process with deterministic and exponential time horizon fixed at the same mean parameter. The functions are reduced by $u$ to isolate the expected gain realized by the pool manager. We can see that the exact formula falls nicely within the 95\% confidence interval bounds of the MC simulations within fixed or exponential time horizon. The red line represents the upper limit of the function to which it converges as $u \to +\infty$, which is also the expected value of the gain in absence of ruin considerations. One can see that for small levels of initial capital potential ruin affects the resulting profit considerably, and for any given $u$ the pool manager can quantify the undesirable effect of ruin.  \\
Figure \ref{fig_psi} exhibits the corresponding ruin probability $\hat{\psi}(u,t)$ for the mining pool. We can note that ruin is highly non-negligible for low levels of initial capital.
Indeed, $\hat{\psi}(u,t=336)<5\%$ for $u>22594$, which is equivalent to $\$5238419$. We also see how the exponential time horizon slightly underestimates the ruin probability for low capital levels, which is due to the skewness of the exponential distribution.\\

\begin{figure}[!ht]
		\includegraphics[width = 350px]{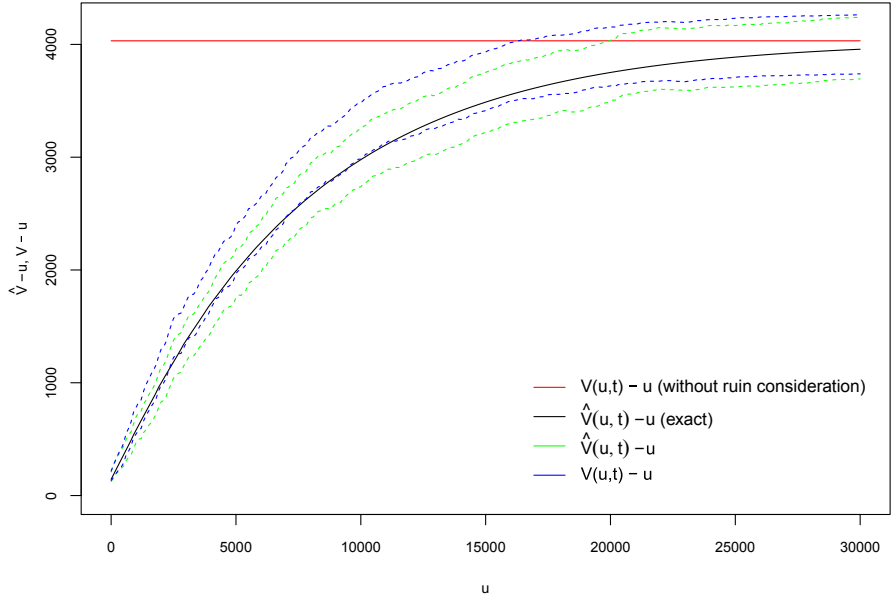}
		\captionsetup{width=0.8\textwidth}
		\centering
		\caption{$\hat{V}(u,t)-u$ as a function of $u$ and simulated $\hat{V}(u,t)-u$ and $V(u,t)-u$ with their 95\% confidence interval bound in dashed with deterministic size jumps $b$ and $w$.}
		\label{fig_V}
\end{figure}
\begin{figure}[!ht]
		\includegraphics[width = 300px]{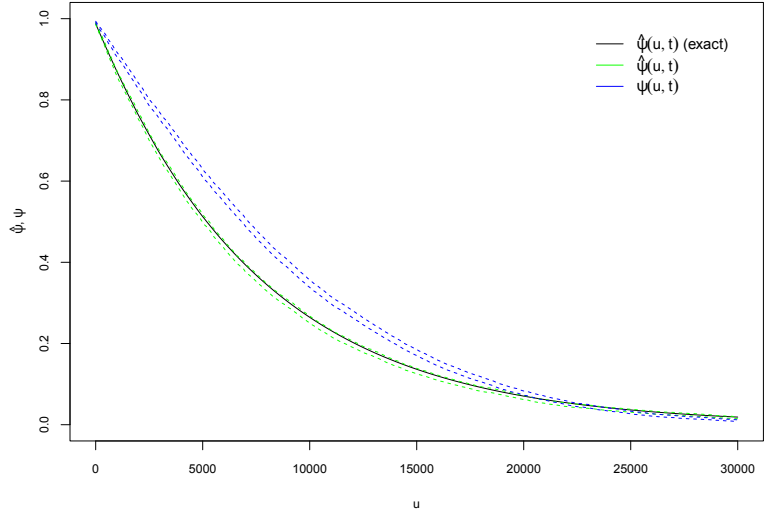}
		\captionsetup{width=0.8\textwidth}
		\centering
		\caption{$\hat{\psi}(u,t)$ as a function of $u$ and simulated $\hat{\psi}(u,t)$ and $\psi(u,t)$ with their 95\% confidence interval bound in dashed with deterministic size jumps $b$ and $w$.}
		\label{fig_psi}
\end{figure}

In Figure \ref{fig_f_deterministic}, we depict the sensitivity of the expected surplus and the ruin probability to the management fee $f$. Not surprisingly, the relationship between $f$ and $\hat{\psi}(u,t)$ is decreasing, as the pool retains more reward for itself. The parameter $f$ impacts the expected gain of the pool manager. 
\begin{figure}[!ht]
	\centering
	\captionsetup{width=0.8\textwidth}
	\begin{subfigure}{0.49\textwidth}
		\includegraphics[width = \textwidth]{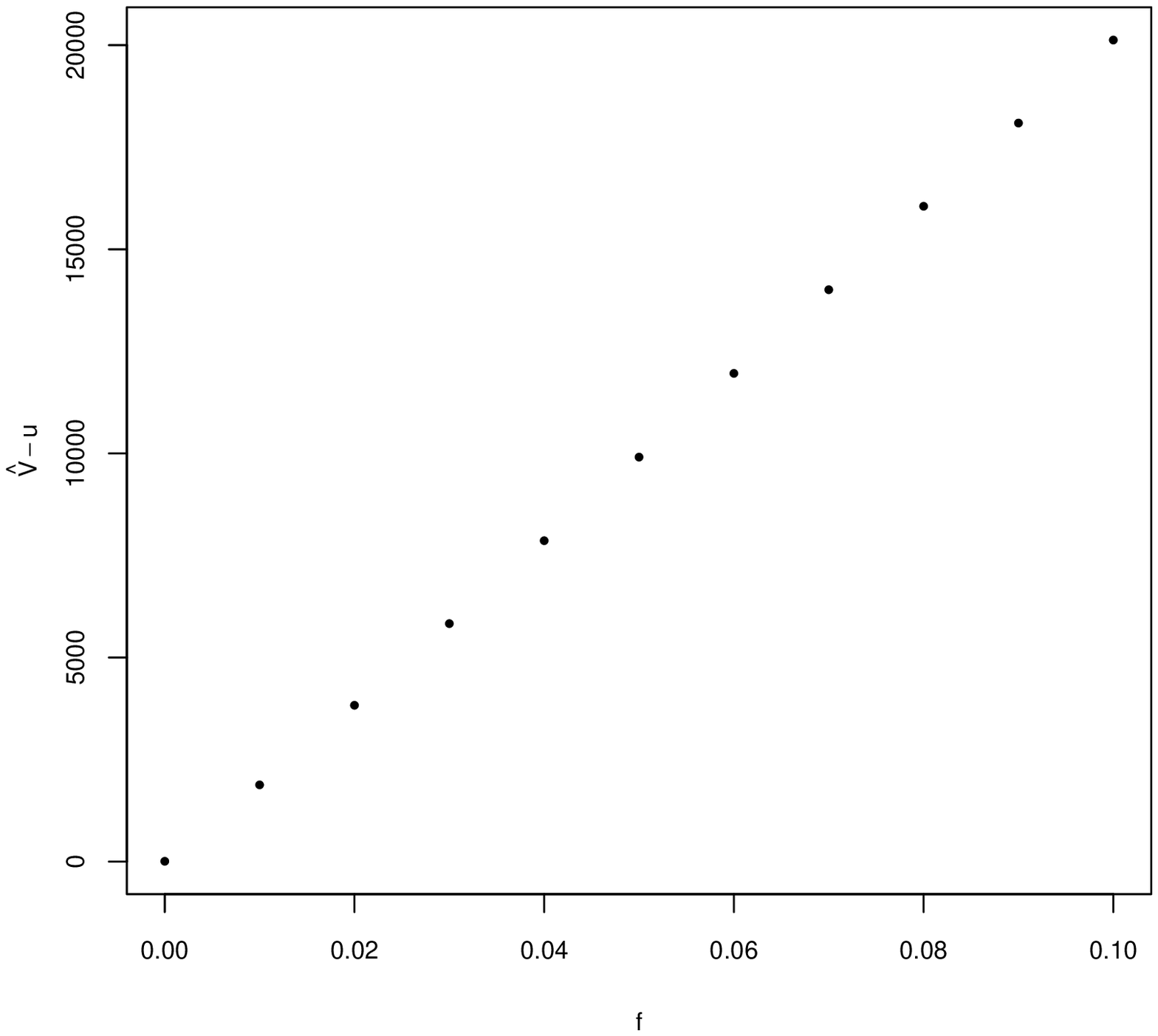}
		\caption{$\hat{V}(u,t)-u$ as a function of $f$.}
	\end{subfigure}
	\begin{subfigure}{0.49\textwidth}
		\includegraphics[width = \textwidth]{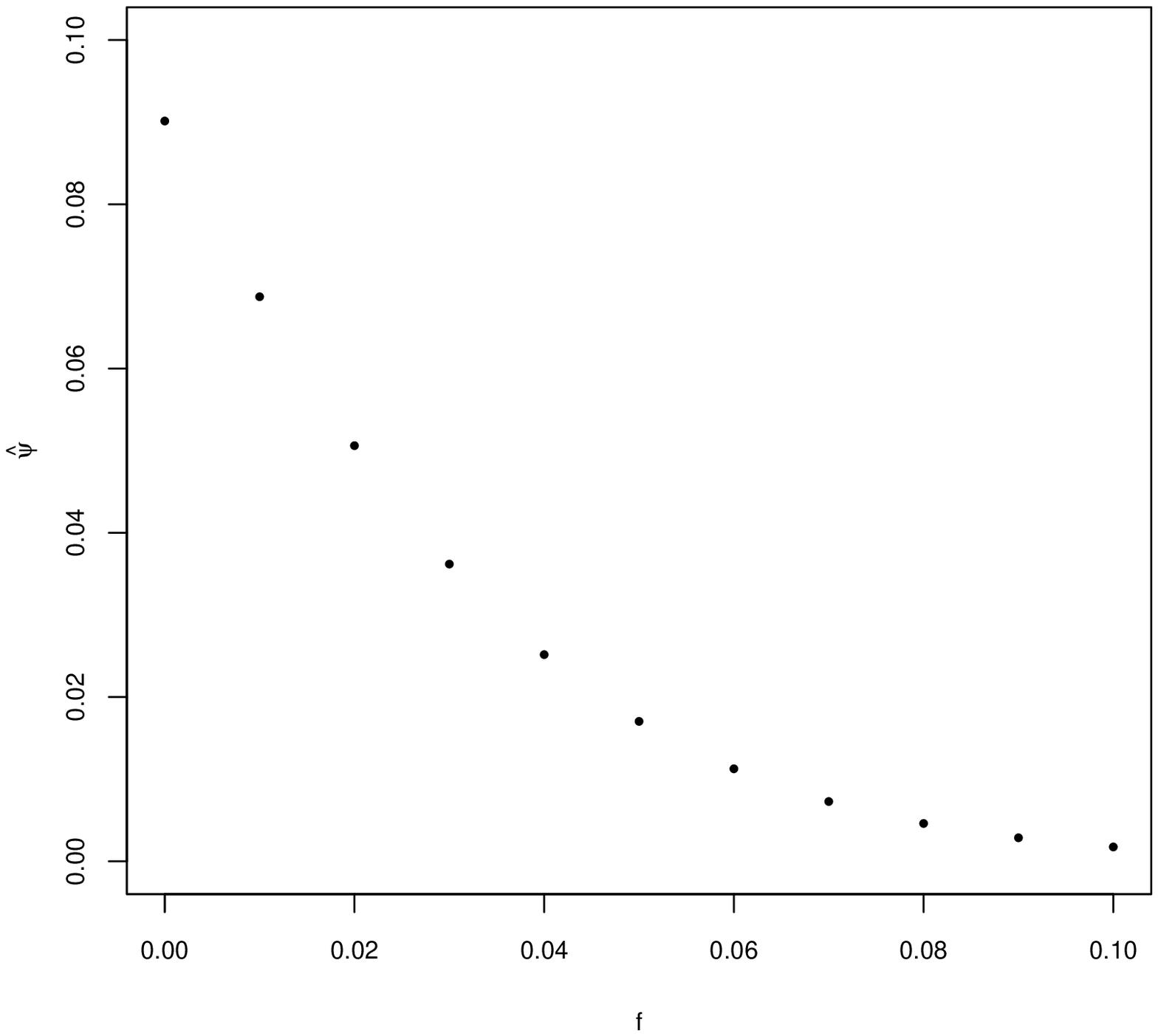}
		\caption{$\hat{\psi}(u,t)$ as a function of $f$.}
	\end{subfigure}
	\caption{Sensitivity to $f$ in case of deterministic rewards and exponential time horizon.}
	\label{fig_f_deterministic}
\end{figure}

In Figure \ref{fig_q_deterministic}, we explore the impact of the relative difficulty to find a share $q$ on ruin and expected surplus. 

\begin{figure}[!ht]
	\centering
	\captionsetup{width=0.8\textwidth}
	\begin{subfigure}{0.49\textwidth}
		\includegraphics[width = \textwidth]{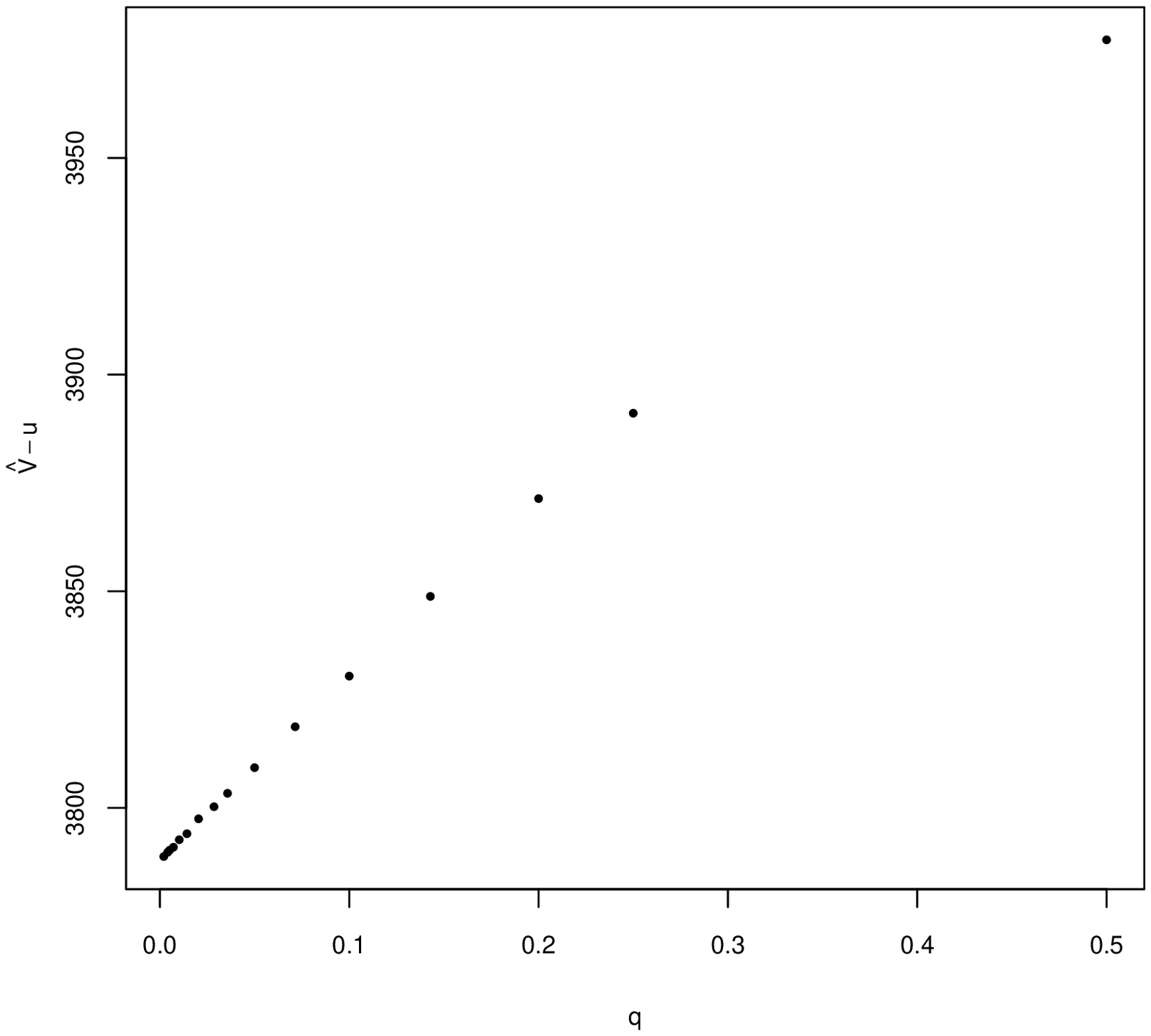}
		\caption{$\hat{V}(u,t)-u$ as a function of $q$.}
	\end{subfigure}
	\begin{subfigure}{0.49\textwidth}
		\includegraphics[width = \textwidth]{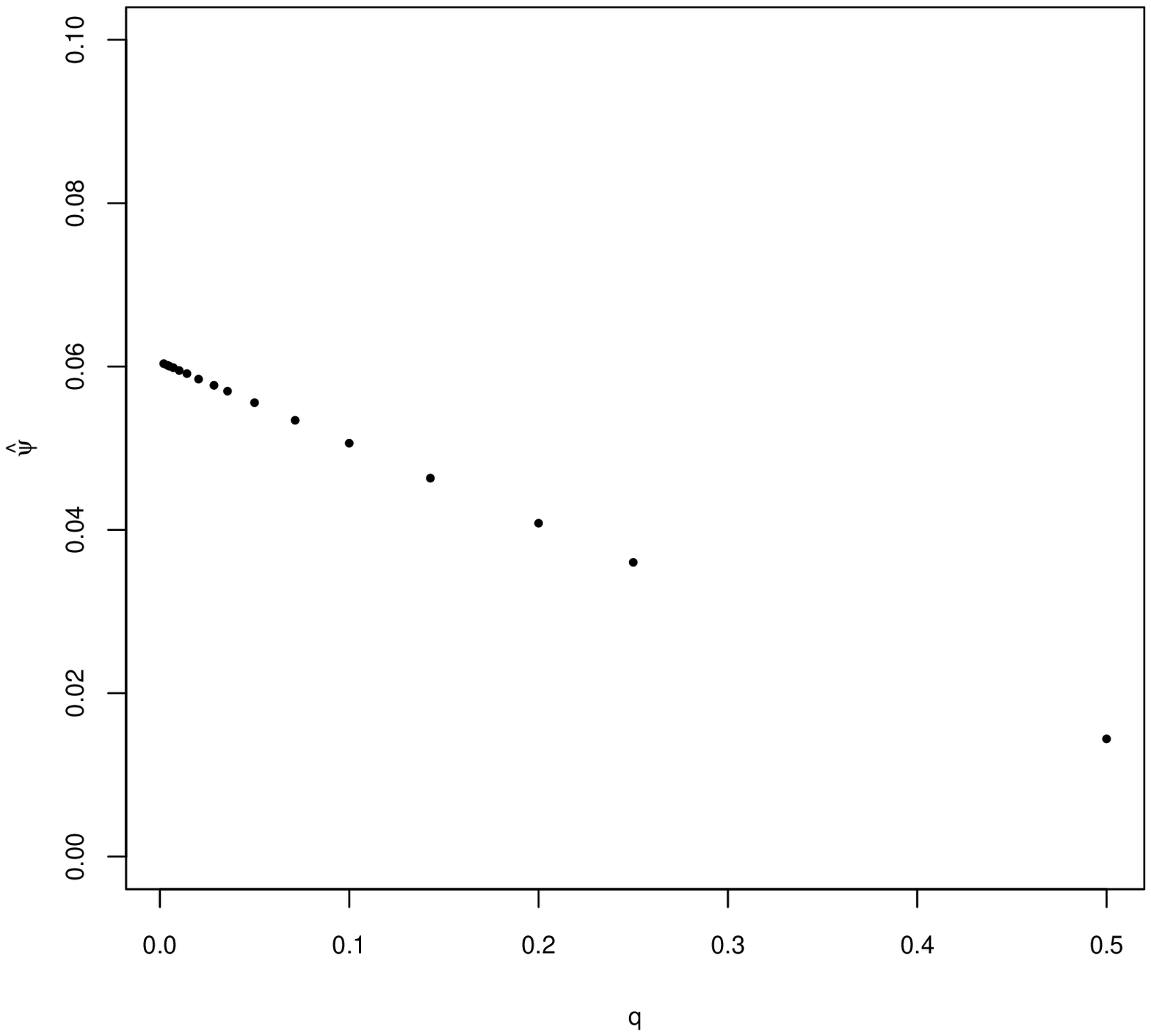}
		\caption{$\hat{\psi}(u,t)$ as a function of $q$.}
	\end{subfigure}
	\caption{Sensitivity to $q$ in case of deterministic rewards and exponential time horizon.}
	\label{fig_q_deterministic}
\end{figure}
It is worthwhile to note that increasing $q$ is profitable to the pool manager. Indeed, as $q$ increases, the payout of shares to the pool members is getting less frequent, thus the pool manager retains more liquidity and controls his probability of ruin at lower levels. The parameter $q$ adjusts the magnitude of the risk transfer between the miners and their manager. 

Figures \ref{fig_V_exp}, \ref{fig_psi_exp}, \ref{fig_f_exponential}, \ref{fig_q_exponential} illustrate the same concepts with exponentially distributed rewards. For comparison, the parameters for the exponential distributions are chosen so that the resulting mean matches the deterministic jump sizes, i.e. $\alpha = 1/w = 1/98, \beta = 1/b = 1/1000$.

\begin{figure}[!ht]
		\includegraphics[width = 300px]{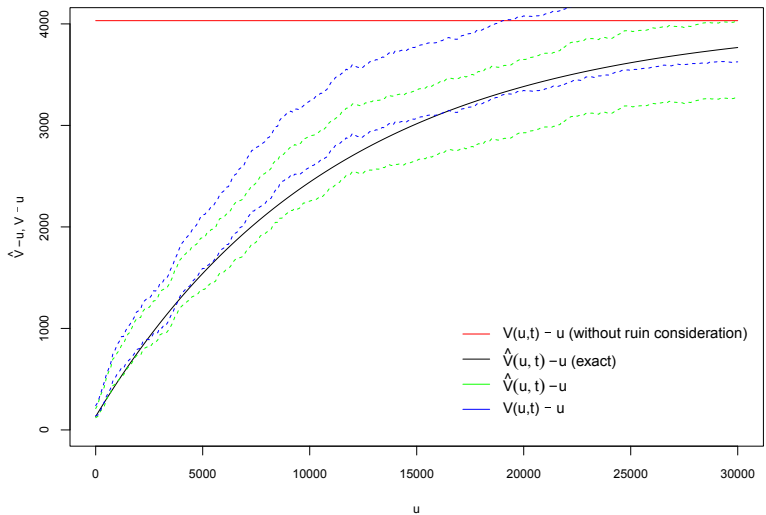}
		\captionsetup{width=0.8\textwidth}
		\centering
		\caption{$\hat{V}(u,t)-u$ as a function of $u$ and simulated $\hat{V}(u,t)-u$ and $V(u,t)-u$ with their 95\% confidence interval bound in dashed. Both jumps are exponentially distributed.}
		\label{fig_V_exp}
\end{figure} 
\begin{figure}[!ht]
		\includegraphics[width = 300px]{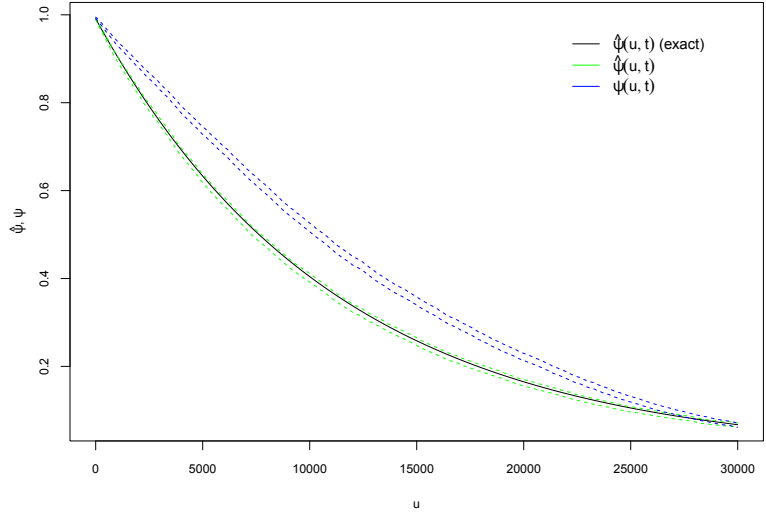}
		\captionsetup{width=0.8\textwidth}
		\centering
		\caption{$\hat{\psi}(u,t)$ as a function of $u$ and simulated $\hat{\psi}(u,t)$ and $\psi(u,t)$ with their 95\% confidence interval bound in dashed. Both jumps are exponentially distributed.}
		\label{fig_psi_exp}
\end{figure}

\begin{figure}[!ht]
	\centering
	\captionsetup{width=0.8\textwidth}
	\begin{subfigure}{0.49\textwidth}
		\includegraphics[width = \textwidth]{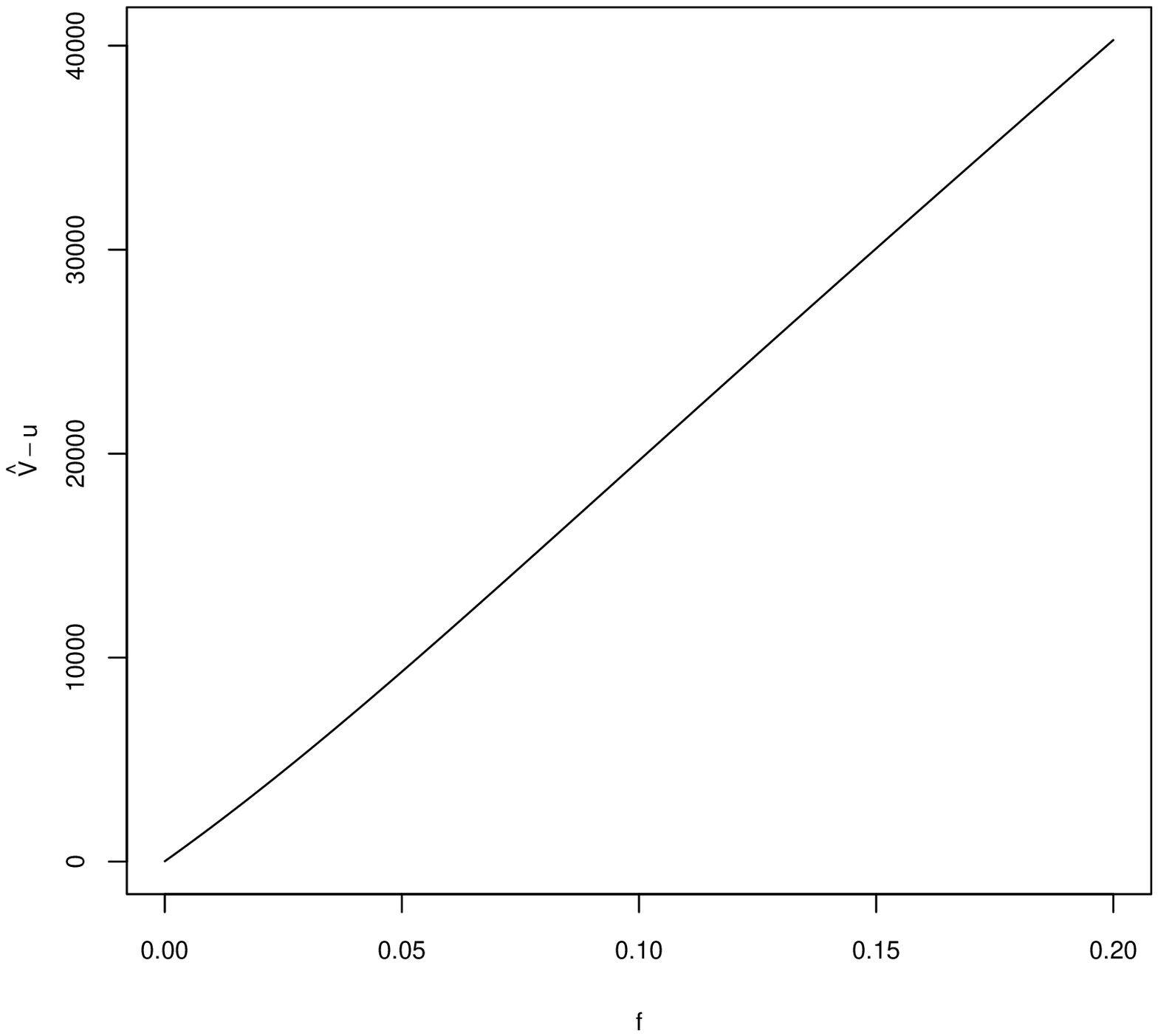}
		\caption{$\hat{V}(u,t)-u$ as a function of $f$.}
	\end{subfigure}
	\begin{subfigure}{0.49\textwidth}
		\includegraphics[width = \textwidth]{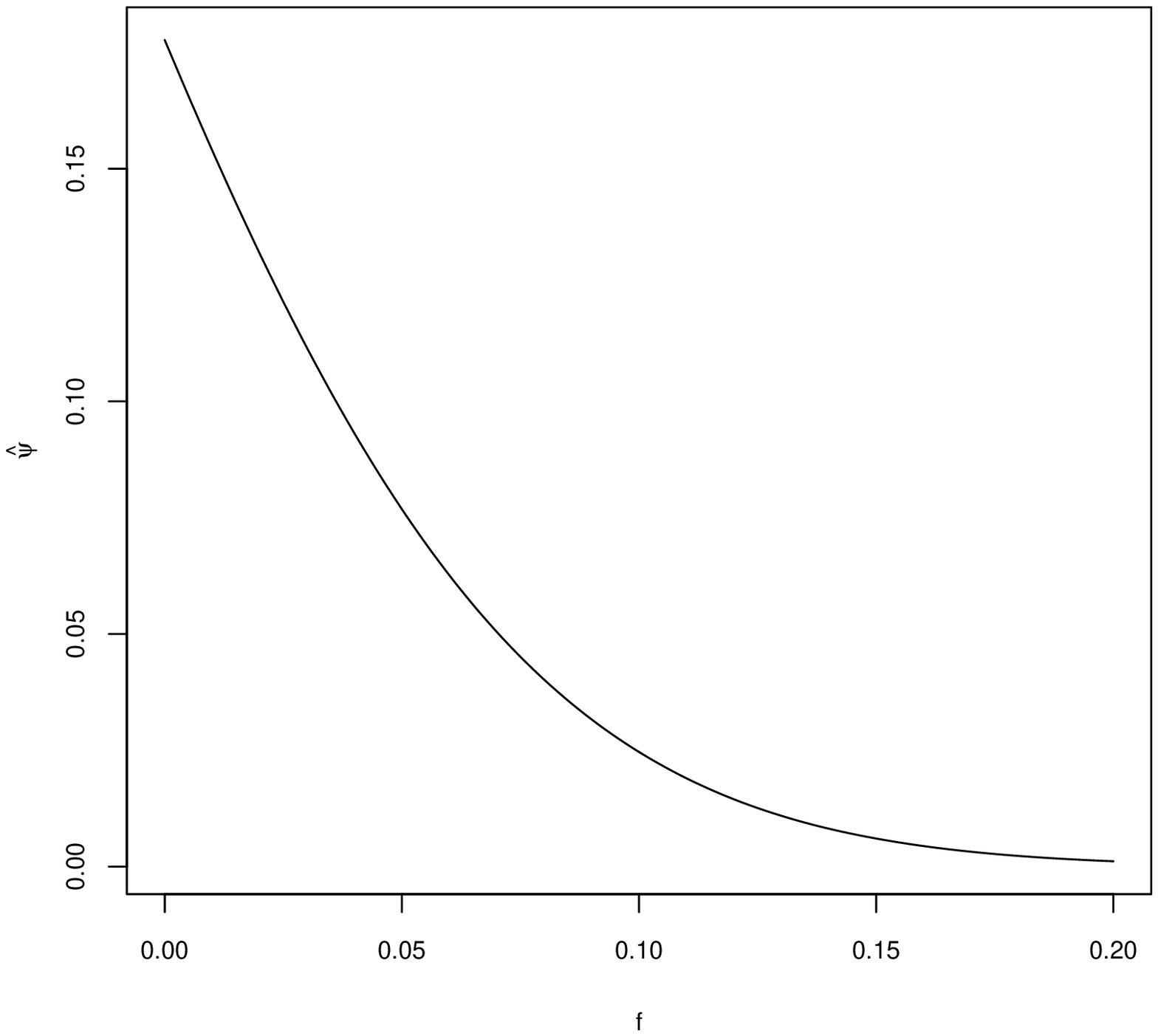}
		\caption{$\hat{\psi}(u,t)$ as a function of $f$.}
	\end{subfigure}
	\caption{Sensitivity to $f$ in case of exponentially distributed rewards and exponential time horizon.}
	\label{fig_f_exponential}
\end{figure}

\begin{figure}[!ht]
	\centering
	\captionsetup{width=0.8\textwidth}
	\begin{subfigure}{0.49\textwidth}
		\includegraphics[width = \textwidth]{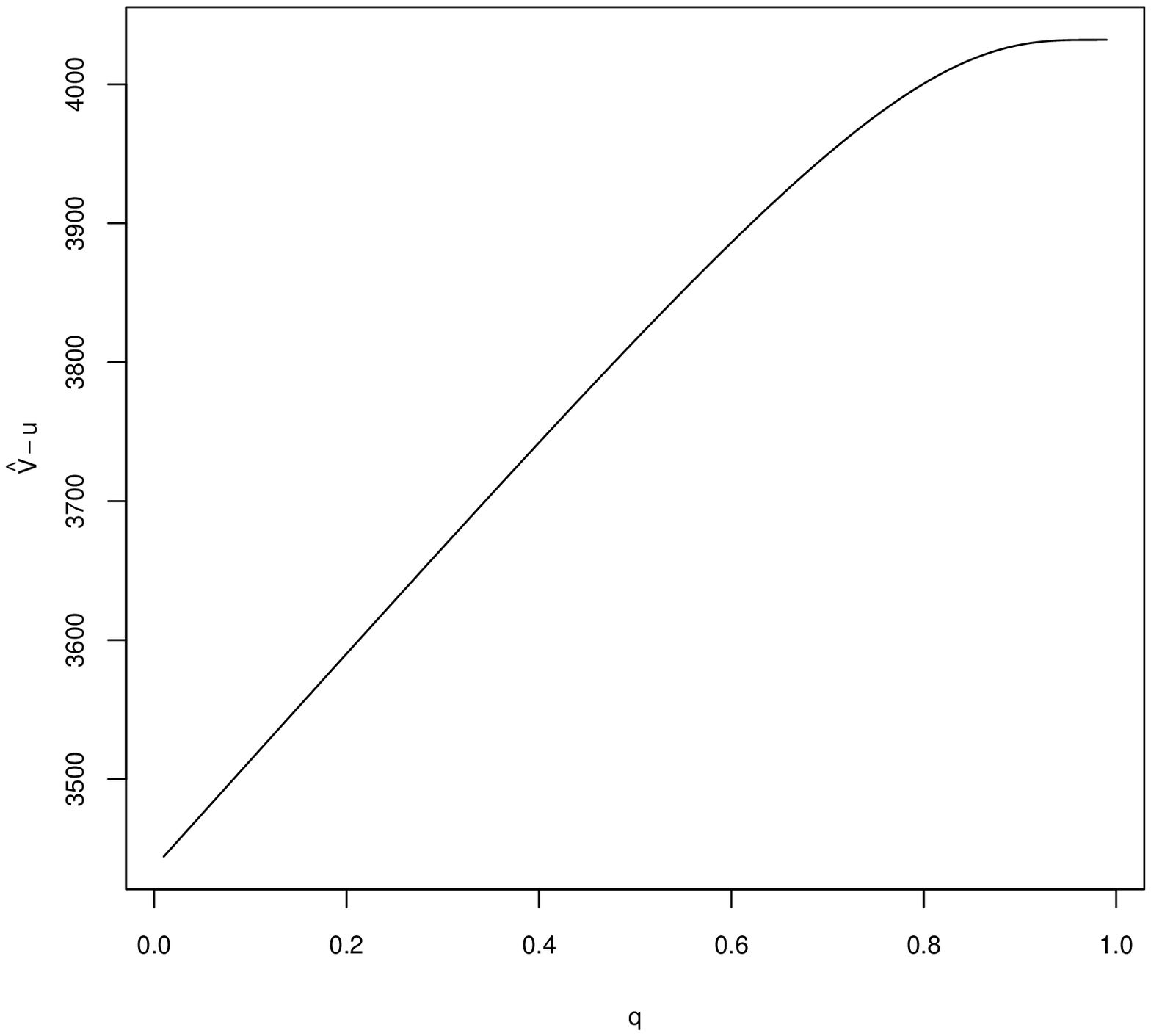}
		\caption{$\hat{V}(u,t)-u$ as a function of $q$.}
	\end{subfigure}
	\begin{subfigure}{0.49\textwidth}
		\includegraphics[width = \textwidth]{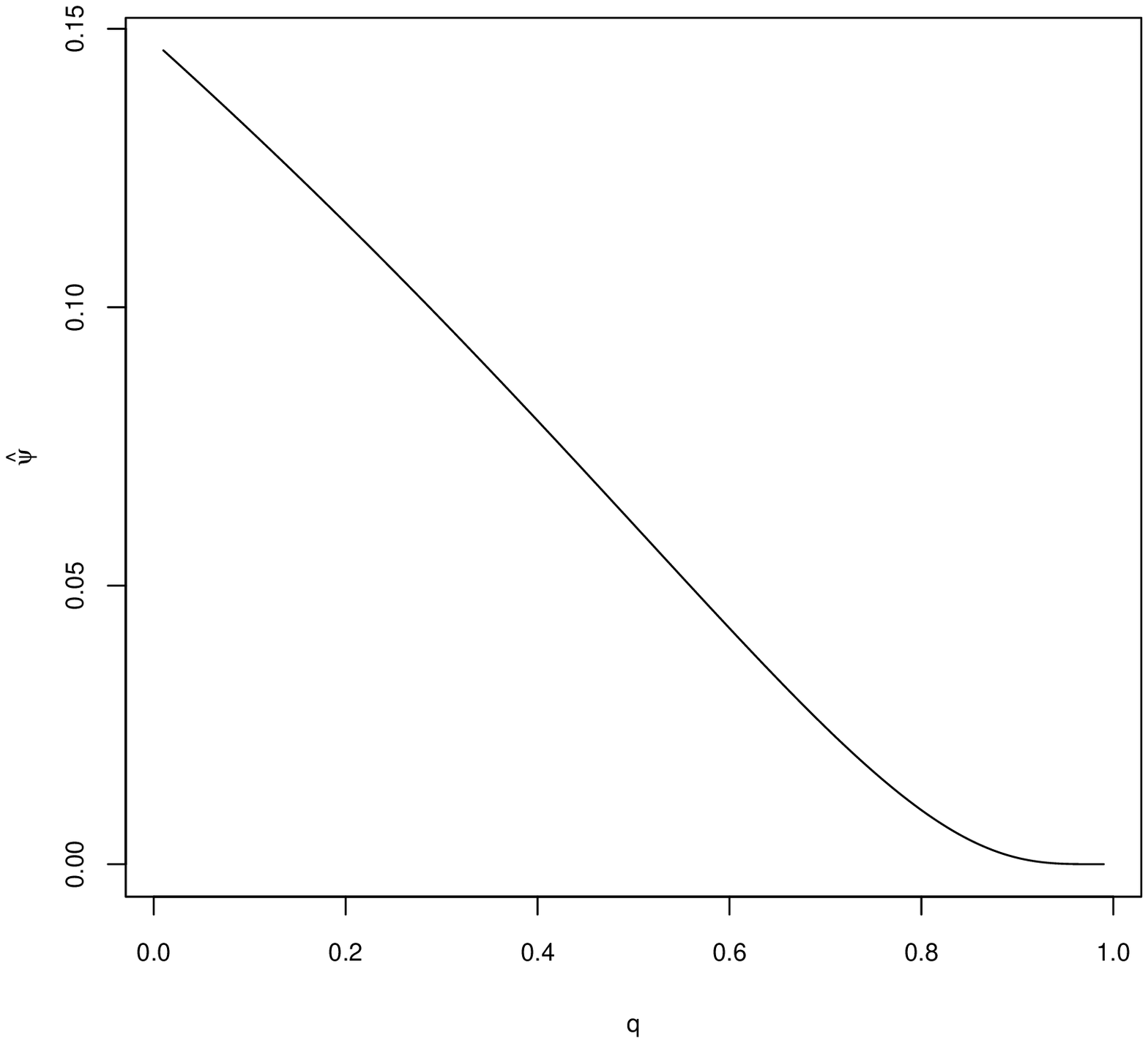}
		\caption{$\hat{\psi}(u,t)$ as a function of $q$.}
	\end{subfigure}
	\caption{Sensitivity to $q$ in case of exponentially distributed rewards and exponential time horizon.}
	\label{fig_q_exponential}
\end{figure}

Figure \ref{fig_V_p_f_exponential} gives a two-way sensitivity analysis with respect to the pool size $p_I$ and the pool fee $f$. The level curves indicate the expected profit for the pool manager for different pool sizes. 

\begin{figure}[!ht]
    \includegraphics[width = 300px]{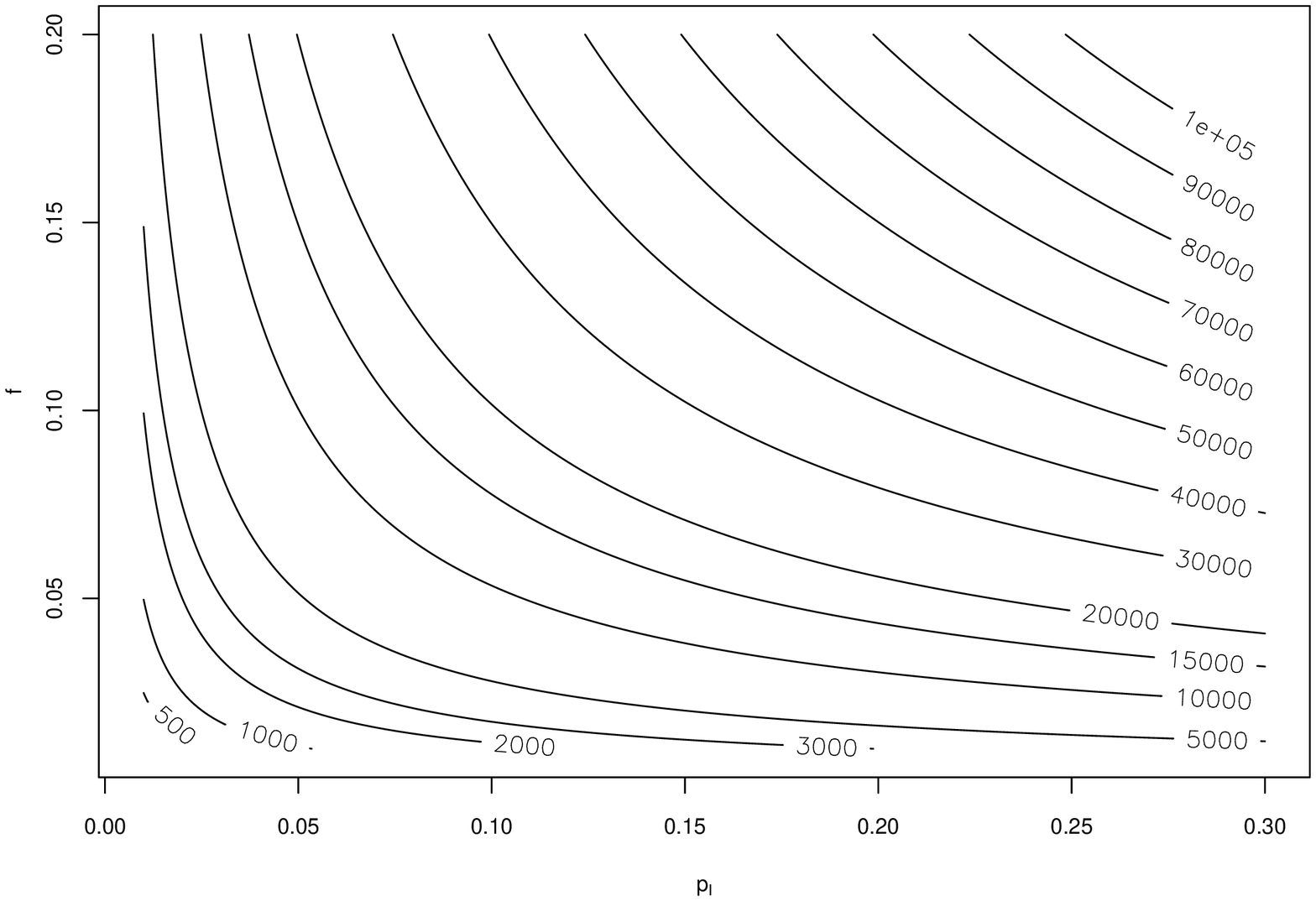}
    \captionsetup{width=0.8\textwidth}
    \centering
    \caption{$\hat{V}(u,t)$ as a function of $p_I$ and $f$ for $u=22500$. Both jumps are exponentially distributed.}
    \label{fig_V_p_f_exponential}
\end{figure}
For a bigger pool size $p_I$, in order to maintain the same level of expected profit, the pool manager can reduce the fee size. One can clearly see an inverse relationship between the pool size and the fee. Thus, a bigger pool can diminish its fees to attract more miners and thus to grow even more. This implies a threat on the decentralized nature of the consensus protocol. If a mining pool manager concentrates more than $50\%$ of the total hashpower, then the blockchain is prone to $51\%$-type attacks such as double spending in the bitcoin context. How can a smaller mining pool tackle this problem? One solution consists in offering to take on more risk by decreasing the difficulty of finding a share which reduces to decreasing the value of $q$. Figure \ref{fig_q_f_exponential_Small_Big_miner} shows the expected profit of two mining pools, one for which $p_I=0.1$ and a smaller one for which $p_I = 0.02$, both having an initial capital level $u=22500$, for both the reward and the time horizon being exponentially distributed. 
\begin{figure}[!ht]
  \centering
  \captionsetup{width=0.8\textwidth}
  \begin{subfigure}{0.49\textwidth}
    \includegraphics[width = \textwidth]{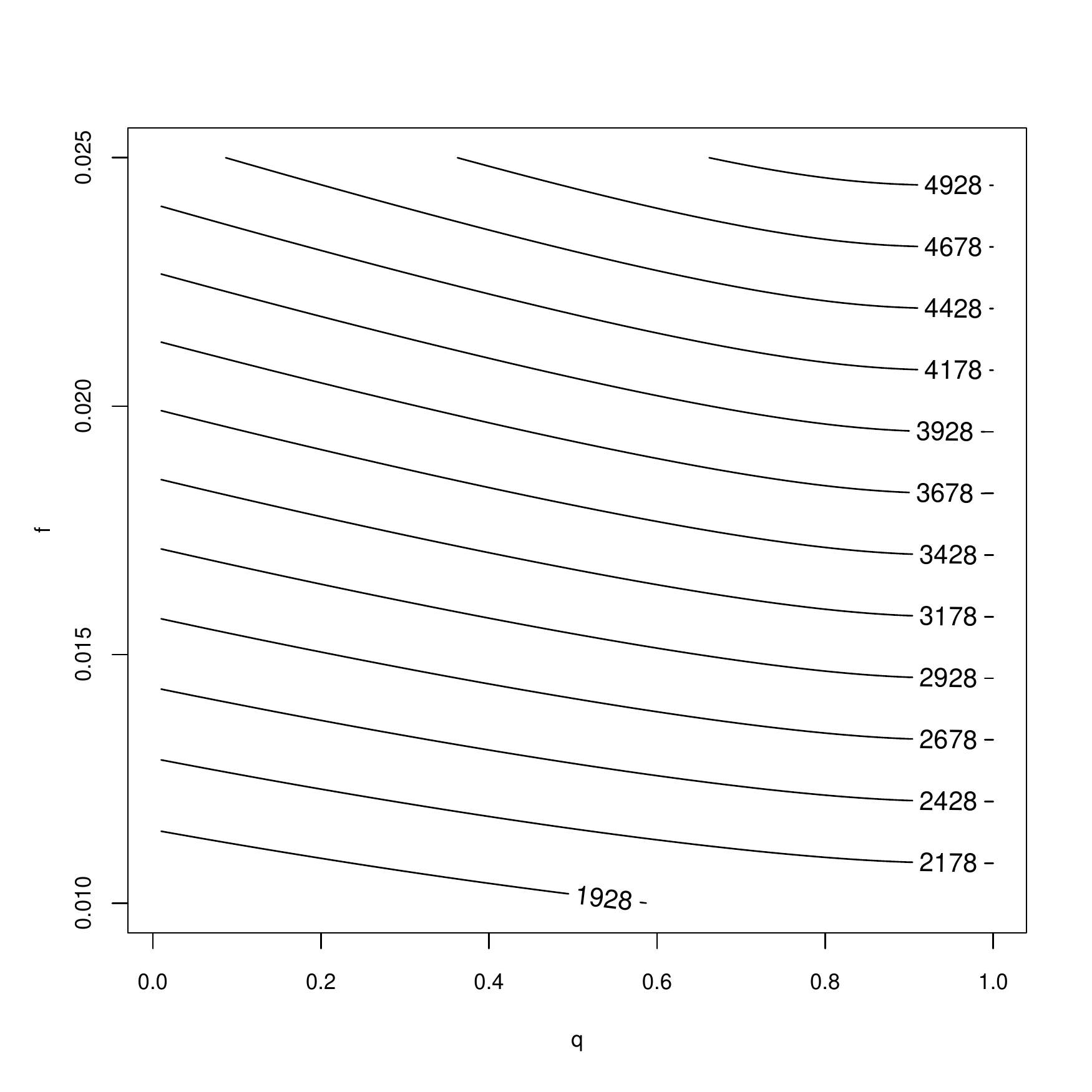}
    \caption{$\hat{V}(u,t)-u$ as a function of $q$ and $f$ for a large mining pool ($p_I = 0.1$).}
  \end{subfigure}
  \begin{subfigure}{0.49\textwidth}
    \includegraphics[width = \textwidth]{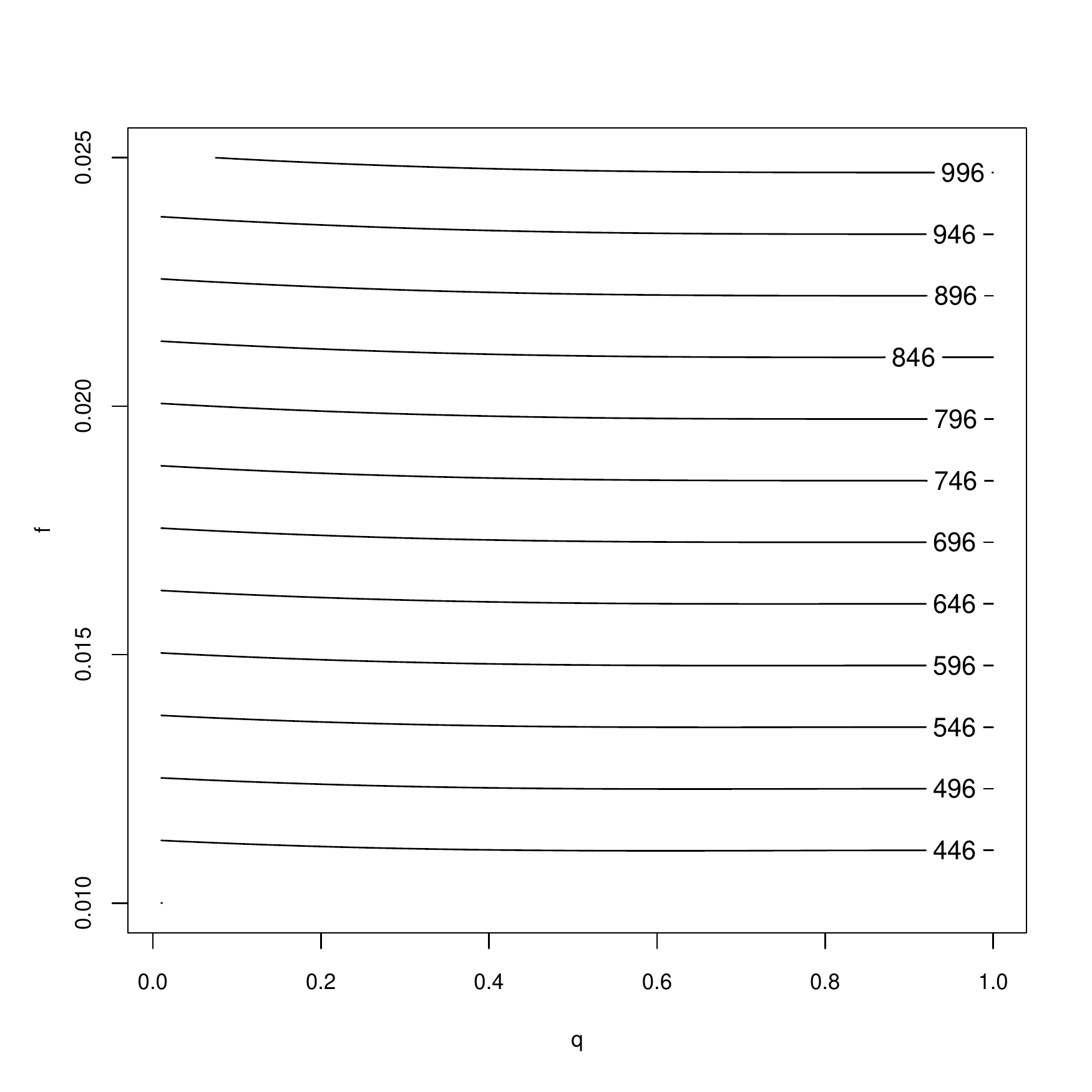}
    \caption{$\hat{V}(u,t)-u$ as a function of $q$ and $f$ for a small mining pool ($p_I = 0.02$).}
  \end{subfigure}
  \caption{Sensitivity to $q$ and $f$ of the expected profit of two mining pools of different size over and exponentially distributed time horizon and reward for $u= 22500$.}
  \label{fig_q_f_exponential_Small_Big_miner}
\end{figure}
The level curves indicate that in terms of expected profit a smaller miner may decrease $q$ without increasing the pool fee $f$, while maintaining the same level of profitability. That is not the case for the larger mining pool whose expected profit turns out to more sensitive to $q$.
\FloatBarrier

\subsection{Individual miner}

Let us now compare the situation of an individual miner before and after joining the pool. We recall Figure \ref{surplus_process} (left panel), which examplifies the pool members' surplus. Also, the surplus of the member is described by \eqref{eq:surplus_miner_fpps}. Finally, we use the results presented in Sections \ref{sec:minerdet} and \ref{sec:minersto} to assess the pool effect for the individual miner's surplus following the protocol. Consider a miner in a deterministic rewards environment. We assume a PPS pool and consider a pool member whose hashpower is equal to 1\% of the pool's total hashpower, i.e. $p_i = 0.001$. For the choice of other parameters, we assume that the cost of electricity $c$ is given by
\begin{equation*}
c = p_i\times W\times \pi_W,
\end{equation*}
where $W$ is the electricity consumption of the network expressed in kWh, and $\pi_W$ is the price of electricity per kWh. For the sake of our example, we take the estimate of $W$ as $\frac{115.541
\times 10^9}{365.25\times 24}$.\footnote{https://cbeci.org/, consulted on May 28th 2021.} The price of electricity is taken to be \$0.06, then converted to our $MU$. Therefore, the net profit condition is satisfied both with and without joining the pool. Figures \ref{fig_V_hat_minervspool} and \ref{fig_psi_hat_minervspool} illustrate the expected surplus and ruin probability with deterministic rewards and exponential time horizon. One can observe how effective the risk reduction in case of joining the pool is for the individual miner. Figure \ref{fig_psi_hat_minervspool} particularly emphasizes the drastic decrease of ruin probability for low capital levels.

\begin{figure}[!ht]
		\includegraphics[width = 300px]{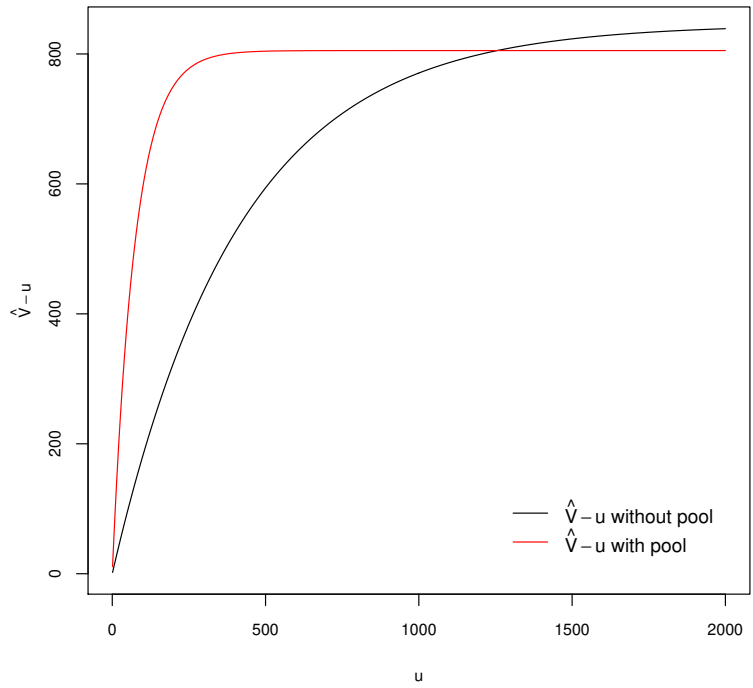}
		\captionsetup{width=0.8\textwidth}
		\centering
		\caption{$\hat{V}(u,t)-u$ as a function of $u$ for an individual pool miner alone in black and within the pool in red.}
		\label{fig_V_hat_minervspool}
\end{figure}
\begin{figure}[!ht]
		\includegraphics[width = 300px]{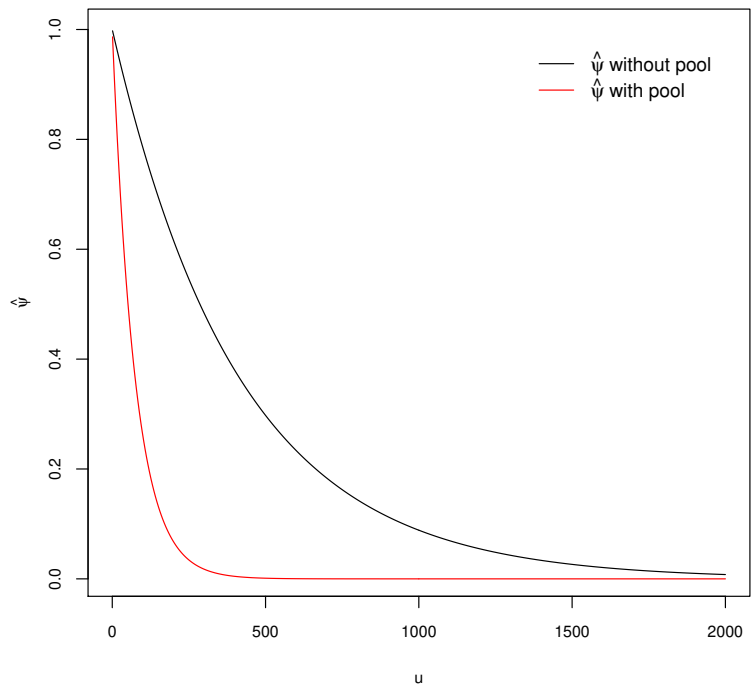}
		\captionsetup{width=0.8\textwidth}
		\centering
		\caption{$\hat{\psi}(u,t)$ as a function of $u$ for an individual pool miner alone in black and within the pool in red.}
		\label{fig_psi_hat_minervspool}
\end{figure} 

Up until a level of initial capital of $u=1255$, it is more profitable for the miner to join the pool, whereas for higher levels of capital the pool fee becomes the main decision driver instead of the ruin considerations. Converted to USD, this amounts to approximately $\$290,971$. Recall that this is akin to the effects of reinsurance, as the miner cedes part of his risk to the pool in exchange of a fixed contractual payment (pool fee).\\

Finally, we investigate the sensitivity of the miner's expected surplus with respect to the key model parameters. In Figure \ref{fig_V_of_f_miner}, the  miner can see for his level of initial capital $u$ whether it is better to join the pool or not, depending on the employed fee $f$. As before, for higher levels of capital, the miner is less willing to accept high fees than a miner with less initial capital. We also observe that the two red lines (miner in the pool with different initial capital $u$) are much closer to each other than the two black lines (miner outside of the pool with different initial capital $u$). This is due to the risk reduction of the miner inside the pool, since he is transferring part of the risk to the pool and getting more frequent rewards. 

\begin{figure}[!ht]
		\includegraphics[width = 300px]{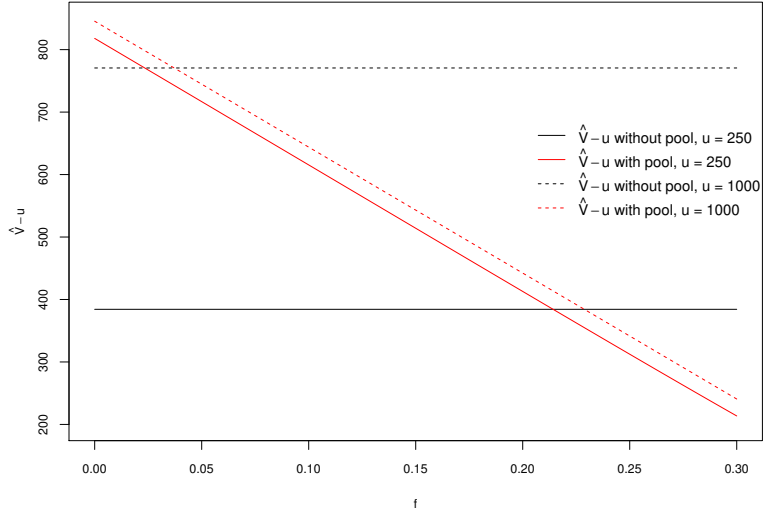}
		\captionsetup{width=0.8\textwidth}
		\centering
		\caption{$\hat{V}(u,t)-u$ as a function of $f$ for an individual pool miner alone in black and within the pool in red.}
		\label{fig_V_of_f_miner}
\end{figure}
Figure \ref{fig_V_of_q_f_miner} shows the level curves of $\hat{V}(u,t)$ with a varying difficulty for the miner's problem $q$ and pool fee $f$. Note that not joining the pool is equivalent to setting the difficulty level equal to the block finding problem level and letting the pool fee be $f=0$. 

\begin{figure}[!ht]
		\includegraphics[width = 300px]{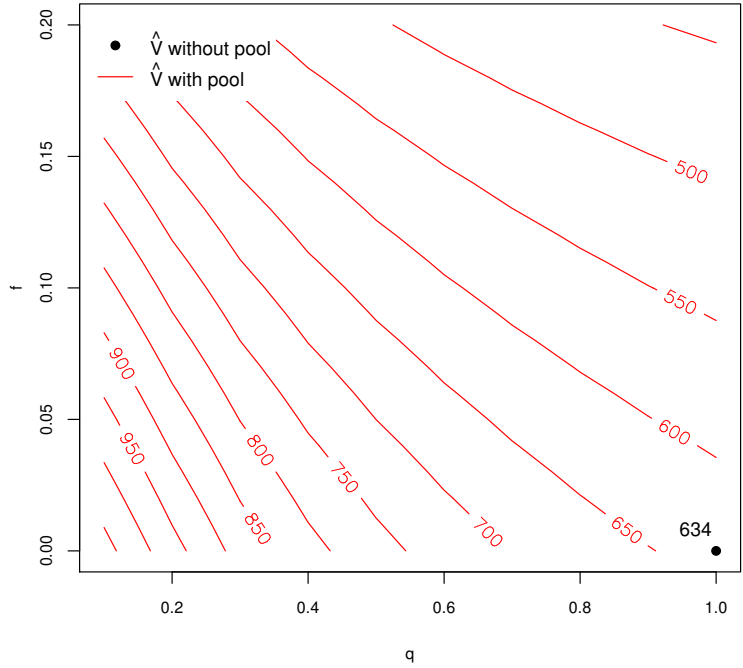}
		\captionsetup{width=0.8\textwidth}
		\centering
		\caption{$\hat{V}(u,t)$ as a function of $q$ and $f$ for an individual pool miner alone in black and within the pool in red.}
		\label{fig_V_of_q_f_miner}
\end{figure}
With such a two-way analysis, the pool can fix an appropriate fee and the miner can see whether he is better off joining the pool for his given level of capital $u$. \\

The miner's decision to join a \textit{pay-per-share} mining pool does not depend on the size of the mining pool. Hence a miner will be indifferent whether to direct her hashpower towards a small or large pool. All that matters is the level of expected profit (decreasing in $f$) and the share of risk transferred to the mining pool (decreasing in $q$). Decentralization will prevail if the preferences, more specifically the risk aversion, of both the pool managers and the individual miners are sufficiently heterogeneous. Note that a situation where a mining pool would control most of the computing power is not desirable for anyone. The blockchain would then be prone to attacks and the associated cryptocurrency would no longer be of value.

\section{Conclusion}\label{secc}
In this paper, we developed a framework for a bitcoin mining pool analysis from a risk and profitability perspective. Given a pay-per-share pooling scheme, we investigated the profitability of a pool under ruin probability considerations, which allows us to derive original results for the pool manager’s expected profit. When describing the pool income process as a stochastic double-sided jump process, one can adapt techniques developed in the actuarial literature for applications in the blockchain universe. In addition, we also looked at the problem from the individual miner’s side, to identify conditions under which it is profitable for her to enter the pool or not.\\

We find that ignoring ruin considerations highly overestimates the expected gain for a pool for small values of initial capital and quantify the required capital level needed for which the ruin aspect becomes negligible. Moreover, we define a trade-off between the main pool defining parameters to set up conditions for optimizing the pool profit for different levels of capital. For an individual miner, pooling has similar effects as a reinsurance treaty for an insurer. We provide a sensitivity analysis that can be helpful for the miner to select the most appropriate pool given his initial parameters.\\

For a randomized time horizon, it was possible to obtain explicit formulas for all quantities of interest. The flexibility of our model enabled to consider deterministic as well as stochastic reward sizes. The established formulas for combinations of exponentials are in fact quite flexible, as any other distribution on the positive halfline can be approximated arbitrary well with such distributions (cf.\ \cite{dufresne2007fitting}). Naturally, some restrictive assumptions were needed to enable the explicit mathematical treatment in this paper, in particular the assumption of independent and identically distributed jump sizes. It will be interesting in future research to look into relaxing these assumptions.\\

The study of the formation of mining pools naturally raises the question of whether they pose a threat to the decentralized nature of blockchain-based applications. We find that the size of the mining pool does not interfere in a miner's decision making process. A miner chooses a mining pool according to the share of risk she wishes to cede and the profit she wishes to make. The preferences of miners and pool managers have been analysed using game theory in Cong et al.\ \cite{Cong2020} and Li et al.\ \cite{li2019mean}. The results of the present paper may serve as concrete risk management tools for miners and pool managers that could also be integrated as value or cost functions within such a game-theoretic approach. 

\section*{Acknowledgements}  
Hans\-j\"{o}rg Albrecher acknowledges financial support from the Swiss National Science Foundation Project 200021\_191984.


\appendix

\section{Abel-Gontcharov polynomials}\label{sec:ag_polynomials}

  Let $U=\{u_i\text{ , }i\geq 1\}$ be a sequence of real non-decreasing numbers. The (unique) family $\{G_{n}(x\vert U)\text{ , }n\geq 0\}$ of \textit{Abel-Gontcharov polynomials} of degree $n$ in $x$ attached to $U$ is defined as follows. Starting with $G_0(x\vert U)=1$, the polynomials $G_n(x\vert U)$ satisfy the differential equations
\begin{equation}\label{eq:DifferentialEquationAGPolynomials}
G_n^{(1)}(x|U)=n\,G_{n-1}(x|{\cal E}U),
\end{equation}
where ${\cal E}U$ is the shifted family $\{u_{i+1}\text{ , }i\geq 1\}$, and with boundary conditions
\begin{equation}\label{eq:Border AGPolynomial}
G_{n}(u_{1}|U)=0, \quad  n\geq 1.
\end{equation}
So, each $G_n$, $n\geq 1$, has the integral representation
\begin{equation}\label{eq:IntegralRepresentationAGPolynomials}
G_{n}(x|U)=n!\int_{u_{1}}^{x}\left[\int_{u_{2}}^{y_{1}}\text{d}y_{2} \ldots \int_{u_{n}}^{y_{n-1}}\text{d}y_{n}\right]\text{d}y_{1}.
\end{equation}
The polynomials $G_n$, $n\geq 1$, can be interpreted in terms of the joint distribution of the order statistics  $(U_{1:n},\ldots,U_{n:n})$ of a sample of $n$ independent uniform random variables on $(0,1)$. Indeed, for $0\leq x \leq u_1 \leq \ldots \leq u_{n} \leq 1$, we have that
\begin{equation*}\label{eq:ProbabilisticInterpretationAGPolynomials}
P[U_{1:n} \leq u_1, \ldots, U_{n:n} \leq u_{n}\, \mbox{ and }\, U_{1:n}\geq x]
= (-1)^n \, G_n(x \vert u_1, \ldots, u_{n}).
\end{equation*}
This last identity is used inside the proof of Theorem \ref{th_f_ruin_time} together with the following property
Note that
\begin{eqnarray}
G_{n}(x|a+bU) = b^n G_n\left((x-a)/b \, |U\right), \quad n\geq 1, \label{eq:LinearTransform}
\end{eqnarray}
Lastly, the numerical evaluation of \eqref{eq:pdf_ruin_time_deterministic_time_horizon} can rely on the recursive relations
\begin{equation}\label{eq:RecursionAGPolynomials}
G_{n}(x|U)=x^{n}-\sum_{k=0}^{n-1} \binom{n}{k} u_{k+1}^{n-k} G_{k}(x|U), \quad n\geq 1.
\end{equation}
Formula \eqref{eq:RecursionAGPolynomials} follows from an Abelian expansion of $x^n$ based on \eqref{eq:DifferentialEquationAGPolynomials}, and \eqref{eq:Border AGPolynomial}.

\section{Proof of Theorem 3.1}\label{sec:proof_deterministic_time_horizon}

The event $\{\tau\in (t,t+dt)\}$ can be viewed conditioned over the values of the process $(N_t)_{t\geq0}$. In other terms,
\begin{equation}
\{\tau\in (t,t+dt)\} = \bigcup_{n=0}^{+\infty} \{\tau\in (t,t+dt)\}\cap\{N_t=n\}.
\end{equation}
We distinguish according to the value of $N_t$. For $N_t=0$, Equation \eqref{eq_ruin_time} can be rewritten as 
\begin{equation}
\tau = \inf\{t\geq 0 ; M_t^d > u/w\},
\end{equation}
which occurs when the $\lceil \frac{u}{w} \rceil ^{th}$ jump of $M_t^d$ occurs at $t$, where $\lceil x \rceil$ denotes the ceiling function. It follows that 
\begin{equation}
\{\tau\in (t,t+dt)\}\cap\{N_t=0\}=\{S_{\lceil \frac{u}{w} \rceil}^d\in  (t,t+dt)\}\cap\{N_t=0\}
\end{equation}
and
\begin{equation}
f_{\tau\mid N_t = 0}(t) = f_{S_{\lceil \frac{u}{w} \rceil}^d}(t),\ t\geq 0.
\end{equation}

In case $N_t\geq 1$, one needs to constrain $\{M_t^d ,\ t\geq 0\}$ so it does not reach $N_{u,s}w/(b-w)+u/(b-w)$ for any time $s<t$ but does so at $t$. Let $(v_n)_{n\geq0}$ is a sequence of integers defined as $v_n = \lceil n(b-w)/w+u/w\rceil$, $n\geq0$. We have
\begin{equation}
\{\tau\in (t,t+dt)\}\cap\{N_t\geq 1\}=\bigcup_{n=1}^{+\infty}\bigcap_{k=1}^n \{T_k \leq S_{v_{k-1}}^d\}\cap\{S_{v_n}^d\in (t,t+dt)\}\cap\{N_t=n\},
\end{equation}
as $M_t^d>\underbrace{N_t}_{=n} (b-w)/w + u/w$ at the time of the fatal jump (and before $t$, $N_t$ reaches each step before the payout process surpasses it). Now
\begin{equation}
\mathbb{P}\left[\{\tau\in (t,t+dt)\}\cap\{N_t\geq 1\} \right] 
= \sum_{n=1}^{+\infty}\mathbb{P}\left[ \bigcap_{k=1}^n \{T_k \leq S^d_{v_{k-1}}\}\cap\{S_{v_n}^d\in (t,t+dt)\} \mid N_t=n\right]\mathbb{P}\left[N_t=n\right].
\end{equation}
By the order statistic property, we get
\begin{equation} \label{eq:app_LTP}
\begin{split}
&\mathbb{P}\left[ \bigcap_{k=1}^n \{T_k \leq S_{v_{k-1}}^d\}\cap\{S_{v_n}^d\in (t,t+dt)\} \mid N_t=n\right]\\
&=\mathbb{P}\left[ \bigcap_{k=1}^n \{U_{k:n} \leq F_t\left(S_{v_{k-1}}^d\right)\}\cap\{S_{v_n}^d\in (t,t+dt)\} \right]\\
&=\mathbb{P}\left[ \bigcap_{k=1}^n \{U_{k:n} \leq F_t\left(S_{v_{k-1}}^d\right)\}\mid S_{v_n}^d\in (t,t+dt) \right]\mathbb{P}\left[S_{v_n}^d\in (t,t+dt)\right]\\
&=\mathbb{E}\left[(-1)^n G_n\left[0\mid F_t\left(S_{v_0}^d\right),\dots,F_t\left(S_{v_{n-1}}^d\right)\right]\mid S_{v_n}^d\in (t,t+dt)\right]\mathbb{P}\left[S_{v_n}^d\in (t,t+dt)\right],
\end{split}
\end{equation}where $(U_{1:n},\dots,U_{n:n})$ denote the order statistics of $n$ i.i.d. unit uniform r.v. and $G_n(.\mid .)$ denote the Abel-Gontcharov polynomials, see Appendix \ref{sec:ag_polynomials} for a short presentation. Now take $F_t(s) = s/t,\ s\leq t$. In virtue of the property \eqref{eq:LinearTransform}, we have
\begin{eqnarray}
G_n\left[0\mid F_t\left(S_{v_0}^d\right),\dots,F_t\left(S_{v_{n-1}}^d\right)\right]&=&G_n\left[0\mid S_{v_0}^d/t,\dots,S_{v_{n-1}}^d/t\right]\nonumber\\
& =& \frac{1}{t^n}G_n\left[0\mid S_{v_0}^d,\dots,S_{v_{n-1}}^d\right].\label{eq:app_ag}
\end{eqnarray}
Inserting that last expression into \eqref{eq:app_LTP} yields
the announced result \eqref{eq:pdf_ruin_time_deterministic_time_horizon}.


\end{document}